\documentclass[ba,preprint]{imsart}

\RequirePackage{amsthm,amsmath,amsfonts,amssymb}
\RequirePackage[]{natbib} 
\RequirePackage[colorlinks,citecolor=blue,urlcolor=blue,backref=page,backref=page]{hyperref}
\RequirePackage[pdftex]{graphicx}

\pubyear{2026}
\arxiv{2606.00231}
\volume{TBA}
\issue{TBA}
\firstpage{1}
\lastpage{1}

\startlocaldefs
\theoremstyle{plain}

\newtheorem{theorem}{Theorem}[section]
\newtheorem{lemma}[theorem]{Lemma}
\newtheorem{corollary}[theorem]{Corollary} 
\theoremstyle{definition}
\newtheorem{definition}[theorem]{Definition}

\theoremstyle{remark}

\endlocaldefs

\usepackage{float}
\usepackage{booktabs}
\usepackage[utf8]{inputenc}
\usepackage[normalem]{ulem}
\usepackage{numprint}
\usepackage{multirow}
\usepackage{bbm}
\usepackage{comment}

\newcommand{\cites}[1]{\citeauthor{#1}'s \citeyearpar{#1}}

\newcommand{\out}{_\text{out}}

\newcommand{\true}{_\text{true}}

\DeclareMathOperator{\E}{E}

\DeclareMathOperator{\KL}{KL}

\begin{document}

\begin{frontmatter}
\title{On Asymptotic Outlier Rejection in Bayesian Mixed Poisson Regression Models Under Extreme Target and Covariate Values}
\runtitle{Bayesian Regression Models of Counts and Extreme Outliers}

\begin{aug}
\author[A]{\fnms{Ilaria}~\snm{Pia}\ead[label=e1]{ilaria.pia@helsinki.fi}\orcid{0009-0008-7983-2163}}
\and
\author[A,B]{\fnms{Jarno}~\snm{Vanhatalo}\ead[label=e2]{jarno.vanhatalo@helsinki.fi}\orcid{0000-0002-6831-0211}}

\address[A]{Department of Mathematics and Statistics, 
University of Helsinki\printead[presep={,\ }]{e1}}

\address[B]{Organismal and Evolutionary Biology Research Programme, 
University of Helsinki\printead[presep={,\ }]{e2}}
\runauthor{I. Pia and J. Vanhatalo}
\end{aug}

\begin{abstract}
Bayesian models are defined to be fully robust against outliers if observations infinitely far from the other data do not influence the posterior. 
In regression models, this entails a need to consider outliers in both target and covariate values.
While in linear regression these cases are interchangeable, as both lead to anomalously large residuals, this symmetry does not apply to generalized linear models.
Importantly, \citet{Hamura:2025} presented sufficient conditions for mixed Poisson count regression models to be robust against infinitely large target values and proposed a mixed Poisson-Rescaled Beta model fulfilling these conditions. 
We continue from their work and study the robustness properties of mixed Poisson regression models with Gaussian latent variables in the presence of outliers in covariates. 
We show that in count regression the symmetry between covariate and target outliers breaks: mixed Poisson models are not robust to outlier covariates even if they were robust to target outliers.
Furthermore, we show that, as a covariate gets infinitely large, the corresponding regression coefficient posterior collapses to a point-mass distribution concentrated around zero.
We hence introduce a novel robustified mixed Poisson model which we denote as $x$-outlier rejective ($x$-OR) mixed Poisson, and demonstrate its theoretical and practical ability to accommodate outliers in the covariates. 
We investigate robustness properties of alternative ($x$-OR) mixed Poisson models in the presence of moderate outliers with simulations and a real world case study and show that the $x$-OR log-$t$ and Rescaled Beta mixed Poisson models reject both types of outliers, leading to improved inference.

\end{abstract}

\begin{keyword}
\kwd{Count Regression}
\kwd{Robust Bayesian Inference}
\kwd{High Leverage Points}
\end{keyword}

\end{frontmatter}

\section{Introduction}

Observations that strongly deviate from the bulk of the other data points, referred to as outliers, can have a large impact on inference and prediction under standard statistical models, leading to biases in the data analysis results. 
For this reason, the development of models that are robust to outliers, capable of automatically detecting, diminishing, and even rejecting, the effect of outliers on inference, has been a long lasting topic in Bayesian statistics \citep{DeFinetti:1961, Neyman:1971,OHagan+Pericchi:2012,Desgagne:2015}. 
The need for robust methods has grown further in recent years due to the increase in the size of datasets and heterogeneity of data sources, which makes the appearance of outliers increasingly more likely.
In this paper, we present a cohesive classification of different types of outliers in count regression and examine how mixed Poisson regression models, analyzed under the Bayesian framework, behave both in the presence of infinitely large outliers and in the presence of finite but anomalous observations arising from either corrupted target values or corrupted covariate values. 
Specifically, we propose a novel class of models (obtained as mixture of mixed Poissons) that is theoretically robust to extreme outliers originating from a corrupted covariate value.

In Bayesian analysis, the presence of outliers generates a conflict between the outlier likelihood component counterpoised to the likelihood of the rest of the data together with the priors \citep{OHagan+Pericchi:2012}. 
The resolution of this conflict is typically expressed in terms of the posterior asymptotic behavior: a Bayesian model is generally claimed to be \emph{fully robust} when the posterior given all the data converges in distribution to the posterior given all but the outlier observations, as the outliers move infinitely far from the bulk of the rest of the data.  
Such posterior convergence is generally referred to as \emph{asymptotic outlier rejection}, as well as \emph{outlier-proneness} \citep[][]{OHagan:1979}.
It should be noted though that the concept of outlier-proneness was first introduced by \cite{Neyman:1971} to refer to the ability of the data generating process to generate extremely large data. As such, in a regression modeling context it only corresponds to robustness against outlying targets, but it cannot be used to refer to robustness against covariates, as these are not considered a product of a data generating process. In this work, we hence define robustness in terms of asymptotic outlier rejection.
Recent studies have provided sufficient conditions for the robustness of linear regression models to outlying target and covariate values \citep{Gagnon:2020, Desgagne:2021, Hamura:2023} as well as for the robustness of generalized linear Gamma \citep{Gagnon:2023} and count regression models \citep{Hamura:2025} to outlying target values.

Specifically, in regression models, an outlying data point arises as an anomalous combination of covariates and target values with respect to the rest of the data. 
Traditional examples are anomalous target values resulting from measurement or reporting errors in the data transcription. However, corrupted covariate observations, also known as high leverage points \citep{Chatterjee:1986}, are also possible. 
For example, in environmental sciences and spatial data analyses, satellite images and raster maps are routinely used to create covariate data.
These products occasionally contain large errors and biases arising from the inaccuracies in the satellite measurements and map production, which may then lead to a mismatch between a spatially referenced target values and covariates associated with them. 
In linear regression, both infinite target and infinite covariate lead symmetrically to infinite residual. We can hence define outliers in terms of their errors, without distinguishing between the source through which the data point presents an anomalous residual \citep{Desgagne:2021}. 
However, as will be shown here, in count regression, the likelihood behavior is not symmetric with respect to an infinite covariate and an infinite target. 
This non-symmetry has, however, been largely ignored in the literature since robustness of generalized linear count models has been studied only in the limit of infinite target values \citep{Hamura:2025}.
Here, we show that this asymmetry between target and covariate outliers has important implications to the posterior inference.
Moreover, in count regression, target values are limited to non-negative values. 
As suggested by \citet{Hamura:2025}, outliers formed by a zero count associated with large intensity could be framed as a zero-inflation problem. While zero-inflated count models have been extensively studied \citep{Cragg:1971,Lambert:1992} and the recent works by \citet{Hamura:2020} and \cite{Hamura:2025} have significantly advanced our understanding of the treatment of large outlying counts, we still lack a comprehensive understanding of how these models behave with different types of outliers originating from anomalous covariate values.

In this work, we classify outliers in count regression based on the corrupted source: the target (type-$y$) or the covariate (type-$x$).
We further classify outliers as large (L) if the corrupted value is very large, in extreme case $+\infty$, and small (S) if the corrupted value is very small, in the extreme case $-\infty$ for outliers of type-$x$.
Since, as will be shown here, the behavior of count models under type-$x$ outliers differs depending on whether the target value is 0 or positive, we further separate outliers of type-$x$ into two subgroups, based on this division.
We summarize the five alternative types of outliers considered in this work in Table \ref{tab:outlier_types}.

%
We then study asymptotic robustness to different types of outliers for a selection of mixed Poisson regression models with Gaussian-latent variables (GLVs), which are among the most common models used in the study of count data \citep{Rue+etal:2009,Wang+Blei:2018}. 
Together with Poisson and Negative Binomial (or more generally Poisson-Gamma mixtures), which are standard choices when modeling count data, we consider the Poisson-Rescaled Beta \citep[-RSB,][]{Hamura:2025} and the Poisson-log Student's-$t$ \citep[first introduced by][]{Gaver87} mixtures. The latter distribution can be seen as a more flexible form of Poisson generalized linear models (GLMs) with GLVs since a Student's-$t$ distribution implies heteroschedastic Gaussian noise for the Poisson rate parameter in the GLM.
We define robustness to infinitely large outliers of type-$y$ in count regression in asymptotic terms, refer to it as asymptotic outlier rejection, and rely on \cites{Hamura:2025} results to study these models' ability to handle such extreme data points.
Specifically, we test whether the Poisson-mixing distribution is right log regularly varying \citep[LRV, ][]{Desgagne:2015}, to handle large outliers of type-$y$. 
We further extend the definition of asymptotic outlier rejection to outliers of type-$x$, originated by infinite covariates, and provide a proof that mixed Poisson models are not robust to them.

We then show that, in the presence of infinite covariates, the regression coefficient posterior distribution collapses to a point mass distribution at zero, also known as Dirac delta. 
This phenomenon expresses similarity to Lindley's paradox \citep{Lindley:1957}: in the presence of an infinite covariate, the likelihood becomes zero in all the regression parameters' space, except for the case with zero valued corresponding coefficient, indicating the model is unable to properly represent the data. However, the regression coefficients posterior results in an extremely certain estimate, putting all the posterior mass at zero.
Based on this result, we introduce a robustified mixed Poisson model which we denote $x$-outlier rejective ($x$-OR), consisting of a mixture of two mixed Poissons, and show that this type of count model is robust against infinite covariates, independently of the Poisson mixing distribution.

We further study mixed Poisson models' behavior in the presence of moderate outliers, generated by finite but anomalous target or covariates via simulations and a real-world case study. We measure robustness to moderate outliers in terms of the Kullback-Leibler (KL) divergence \citep{Kullback+Leibler:1951} of the posterior given the full dataset with respect to the posterior excluding outliers. 
Despite mixed Poisson models not being asymptotically robust against extreme corrupted covariates, in practice, they do perform much better than the baseline Poisson model for data which include moderate outliers. Generally, a Poisson-Gamma (Negative Binomial) model was equally well able to capture outlying data of type-$x$ as the mixed Poisson models that are robust to extreme large outliers of type-$y$.

The rest of the paper is structured as follows. 
In Section \ref{sec:count_models}, we define the family of Bayesian mixed Poisson models considered in our study (Poisson, Negative Binomial, Poisson-log Student's-$t$, and Poisson-RSB mixture) and check the behavior of the tails of their mixing components. 
In Section \ref{sec:robustness_extreme_outliers}, we define count regression models' robustness to extreme outliers of different types in terms of asymptotic outlier rejection and present sufficient conditions for mixed Poisson with GLVs models to be robust to extreme outliers of type-$y$, following \cites{Hamura:2025}. We further show that, while mixed Poissons are not robust to extreme outliers of type-$x$, independently of the mixing component, our proposed $x$-OR mixed Poisson models are robust to this kind of outlier.
In Section \ref{sec:experiments}, we conduct experiments on the different mixed Poisson GLVs count models to investigate their robustness to moderate outliers of different type.
In Section \ref{sec:kara_sea}, we adopt the alternative mixed Poisson models to study the distribution of polar bears in the Kara Sea with respect to a selection of environmental covariates. The data include one moderate outlier of type-$x$, originating from a corrupted satellite image. We compare the models with and without the outlier and demonstrate that even one such outlier can have large impact on model results. 
We end with a synthesizing discussion and conclusions as well as suggestions for future research to resolve the challenge in Section~\ref{sec:discussion}.

\begin{table}[t]
    \centering
        \caption{The basic reasons for an observation, $\{y, x\}$, to be an outlier in count regression models with continuous and unbounded covariates and the resulting type of outlier, conceptualized by whether the outlying observed value is much smaller (S) or larger (L) than its "true" value $\{y_{\text{true}}, x_{\text{true}}\}$.}
    \begin{tabular}{ll|c|c}
    \toprule
     Corrupted value& Direction of error & Outlier type & Extreme case \\
    \midrule
    Target $y$ : &$y>>y_{\text{true}}$ & L-$y$ & $y\to+\infty$ \\
    Covariate $x$ : &$x >> x_{\text{true}}$ & L-$x;y>0$& $x\to +\infty$ \\
        & & L-$x;y=0$& $x\to +\infty$ \\
        & $x << x_{\text{true}}$  & S-$x;y>0$ & $x\to -\infty$ \\ %
        &   & S-$x;y=0$ & $x\to -\infty$ \\ 
        \bottomrule
    \end{tabular}
    \label{tab:outlier_types}
\end{table}

\section{Mixed Poisson count regression}\label{sec:count_models}
     
\subsection{Mixed Poisson regression models with Gaussian latent variables}

In count regression, the standard Poisson model is known to be sensitive to outliers and zero-inflation since the variance and expectation of a Poisson distribution are both equal to the rate parameter.
From a data generating process point of view, we can obtain a more robust distribution, by adding heteroschedasticity to the rate parameter of the Poisson distribution.
This is similar in nature to writing the Student's-$t$ as a scale mixture of Gaussian distributions in linear regression.
Adding a multiplicative varying component to the Poisson rate parameter following a positive-valued mixing distribution leads to a mixed Poisson distribution which is more flexible than a standard Poisson distribution. However, whether a mixed Poisson distribution is robust to outliers or not depends on the mixing distribution and the type of an outlier.

Hereafter we denote by $\mathcal{D}_{n+1} = \{\boldsymbol{d}_1,\dots, \boldsymbol{d}_{n+1}\}$ a set of regression data such that $\boldsymbol{d}_i = (y_i,\boldsymbol{x}_i)$, with $y_i\in \mathbb{N}_0$ being the target and $\boldsymbol{x}_i\in \mathbb{R}^p$ the vector of covariates for observation $i=1,\dots,n+1$. 
We assume observations are conditionally independent.
We consider the target expectation model $\E(y_i|\mathcal{D}_{\setminus i}, \boldsymbol{x}_{i}) = \lambda_i\eta_i$ with $\lambda_i = \exp(\textbf{x}_i^\top\boldsymbol{\beta} + \zeta_i)$, generalized regression model with logarithmic link function, where $\boldsymbol{\beta}$ is the $p$-dimensional vector of regression coefficients, with Gaussian prior, and $\zeta$ is a $q$-dimensional Gaussian random effect capturing 'regular' observations, while $\eta_i$ is a multiplicative random effect capturing outliers.
Specifically, we consider mixed Poisson regression models, such that
 \begin{align}\label{eq:Pois-mix-models}
     y_i|\lambda_i\eta_i&\sim \text{Poisson}(\lambda_i\eta_i)& \nonumber\\
     \eta_i&\sim \pi (\cdot) &\\
     \lambda_i& = \exp(x_i^\top\beta + \zeta_i) : \beta, \zeta_i\sim N(\cdot).& \nonumber
 \end{align}
Here, the mixed Poisson rate parameter is attained as the product of a noise coefficient $\eta_i\sim\pi$ and a rate parameter $\lambda_i>0$ that follows a log-Gaussian distribution. 
Assuming log Gaussian distribution for $\lambda$ is a common choice in count regression models. 
Note though, that all following results can be extended to mixed Poisson regression whose regression components follow any bounded continuous prior.

\subsection{Alternative mixed Poisson models}

Many mixed Poisson distributions have been proposed in the literature, and their properties have been studied, for example, by \cite{Karlis:2005}.
In this section, we compare four of them: Poisson, Negative Binomial (Poisson-Gamma), Poisson-log Student's-$t$ \citep{Gaver87}, and Poisson-Rescaled Beta (Poisson-RSB) \citep{Hamura:2025}. 
The first two are standard choices in the analysis of count data, the latter two are less common but heavier tailed alternatives.

For each of the mixed models considered, we study the behavior of the left and right tails of the multiplicative noise $\eta_i$. 
We check whether the right tail of the mixing distribution $\pi$ is log regularly varying \citep[LRV,][]{Desgagne:2015} and whether the mixing distribution diverges with, at most, a polynomial rate at the origin: 
\begin{align}
 \text{LRV right tail}: \hspace{0.5cm}&  \pi(\eta) \approx \eta^{-1}(\log \eta)^{-(1+b)}, \ b>0 & \text{as } \ \eta\to \infty, \label{eq:right_tail}  \\ %
 \text{Polynomial divergence at origin}: \hspace{0.5cm}&   \pi(\eta) \approx \eta^{a-1}, \ 0<a<1 & \text{as } \ \eta\to 0.   \label{eq:left_tail}
\end{align}
The symbol $\approx$ here denotes the asymptotic proportionality between functions; $f(x)\approx g(x)$ meaning $f(x)/g(x) = C<\infty$ for $x\to0 \text{ or }\infty$. 
The left and right tails behavior of the different Poisson-mixing distributions is shown in Figure \ref{fig:LRtails}.

\begin{figure}[t]
    \centering
    \includegraphics[width=.8\linewidth]{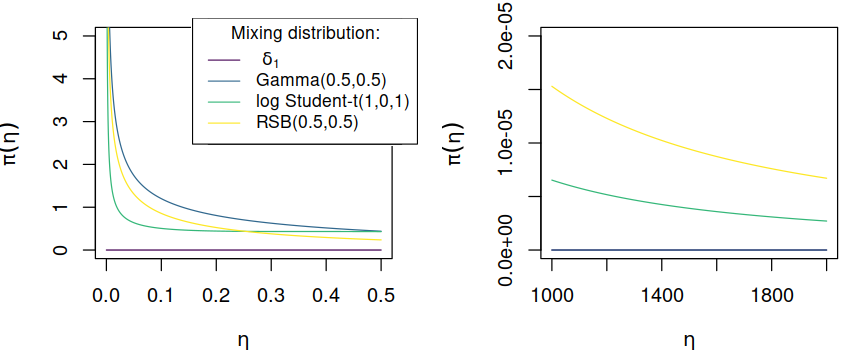}
    \caption{Behavior of the different Poisson-mixing distributions, $\eta_i\sim \pi$ (Section \ref{sec:count_models}), at the origin (left) and at the right tail (right). }
    \label{fig:LRtails}
\end{figure}

\subsubsection{Poisson model}
    
As a baseline model, we consider the Poisson process: $y_i|\lambda_i \sim \text{Pois}(\lambda_i)$. 
Assuming that data are generated from a Poisson process implies that the variability in data equals their average, as the Poisson rate parameter $\lambda$ serves both as a scale and as a location parameter: $\text{Var}[Y]=\E[Y]=\lambda$. Often, data expresses higher variability though. This leads to overdispersion and makes the Poisson distribution non-robust against outliers.

The Poisson model can be seen as a mixed Poisson with multiplicative noise $\pi = \delta_1$, the one point mass at 1. Clearly it is not heavy-tailed nor diverging at origin.

\subsubsection{Negative Binomial (Poisson-Gamma) model}

A common approach to better model overdispersed data is to introduce heteroschedasticity into the data generating process by assigning a Gamma prior for the mixing distribution, so that
\begin{align*}
y_i | \lambda_i, \eta_i & \sim \text{Poisson}(\lambda_i \eta_i) \label{eq_scale_mixture_1}\\
\eta_i|r & \sim \text{Gamma}(r, r). 
\end{align*}
Here, each observation has its own rate parameter $\lambda_i \eta_i$ following a Gamma prior distribution with shape $r>0$ and rate $\frac{r}{\lambda_i}>0$. Marginalizing over $\eta_i$ we obtain that $y_i$ follows a Negative Binomial distribution centered in $\lambda_i$ \citep{Greenwood+Yule:1920}:
\begin{equation*}
    y_i|\lambda_i,r \sim \text{NegBin}(\lambda_i,r).
\end{equation*}
%
The model converges to the Poisson distribution when its overdispersion parameter $r \to \infty$.

In the Negative-Binomial model, the multiplicative noise $\eta_i$ has density
\begin{equation}
    f_{\text{Gamma}}(x;r,r) \propto x^{r-1} e^{-x r}.
\end{equation}
As $x\to\infty$ the density converges to zero at a rate of $e^{-x r}$, hence it is not right LRV. However, as $x\to0$, $f_{\text{Gamma}}(x)\sim x^{r-1}$ which diverges polynomially for $r<1$.
In our experimental studies we set $r=1/2$.

\subsubsection{Poisson-log Student's-$t$ model}

The Poisson-log Student's-$t$ (shortly log-t hereafter) mixture model for count data, first proposed by \citet{Gaver87}, can be derived as follows
\begin{align*}
y_i | \lambda_i, \eta_i & \sim \text{Poisson}(\lambda_i \eta_i) \label{eq_scale_mixture_1}\\
\log\eta_i|k & \sim \text{Student's-}t_k(\cdot). 
\end{align*}
where each observation has its own rate parameter $\lambda_i \eta_i$, and $\log\eta_i$ follows a Student's-$t$ prior distribution with $k$ degrees of freedom. As underlined by \citet{Gaver87} this mixture can be seen as a generalization of the logNormal-mixed Poisson \citep{Bulmer:1974} which allows for a better fit of extreme outliers, since the Student's-$t$ has systematically heavier tails than Gaussian.
Further, the log-t mixed Poisson can be seen as a more flexible form of Poisson GLMs with GLVs, and hence represents a natural choice for a more robust count model.
This model is also equivalent to considering a GLM for the rate parameters $\theta_i = \lambda_i\eta_i$, with additional heteroschedastic Gaussian noise.
Given $y_i\sim\text{Pois}(\theta_i = \lambda_i\eta_i)$, then:

\begin{tabular}{lcl}
\\
    $\log\lambda_i = x_i'\beta + \zeta_i$ & $\iff$ & $\log\theta_i = x_i'\beta + \zeta_i + \log\eta_i$ \\
    $\ \ \zeta_i\sim \text{N}(0,\sigma^2), \ \eta_i\sim \log-t(k,0,s^2)$ && $\ \ \zeta_i\sim \text{N}(0,\sigma^2), \ \log\eta_i\sim \text{N}(0,s_i^2)$\\
    && $\ \ s_i^2\sim \text{Inv}-\chi_k(\cdot)$.\\
\\    
\end{tabular}

The density of a log-t distribution with $k$ degrees of freedom is
\begin{equation}
    f_{\text{log-t}}(x;k) \propto \frac{1}{x}\Big( 1 + \frac{1}{k}\left(\frac{\log x}{s}\right)^2\Big)^{-\frac{k+1}{2}}.
\end{equation}
As $x\to\infty$ the density $f_{\text{log-t}}(x)\sim x^{-1}\log x ^{-(k+1)}$ which is LRV with index $k+1$. The lower the LRV index, the fatter the tail; so setting the log-t degrees of freedom $k=1$ gives the fattest tailed distribution, also known as log-Cauchy distribution. As $x\to0$, $f_{\text{log-t}}(x)$ diverges at the rate of $x^{-1}|\log(x)|^{-k-1}$, slower than the polynomial rate \eqref{eq:left_tail}. 

\subsubsection{Poisson-Rescaled Beta (Poisson-RSB) model}
Following \citet{Hamura:2025}, we derive the Poisson-Rescaled Beta (Poisson-RSB) model as
\begin{align*}
y_i | \lambda_i, \eta_i & \sim \text{Poisson}(\lambda_i \eta_i) \label{eq_scale_mixture_3}\\
\eta_i|a,b & \sim \text{RSB}(a,b), 
\end{align*}
where the parameter $\eta_i>0$ follows a rescaled beta distribution, with density: 
\begin{equation}
    f_{RSB}(x;a,b) = \frac{1}{B(a,b)} \frac{\log(1+x)^{a-1}}{1+x} \frac{1}{[1+\log(1+x)]^{a+b}}. 
\end{equation}
A RSB distributed random variable arises as a scale transformation of a beta-distributed random variable: if $U\sim \text{Beta}(a,b)$ then $X = \log(1 + \frac{U}{1-U})\sim \text{RSB}(a,b)$.
The Poisson-RSB mixture can also be seen as an extension of the Poisson-Gamma mixture since it can be obtained as the marginal distribution of the following hierarchical model: $\eta|r\sim \text{Gamma}(1,r),\ r|u,v\sim \text{Gamma}(u+v,1),\ (u,v)\sim g(u,v)$, bivariate joint density such that $g(v,w)\propto\frac{v^{-a}w^{b+a-1}e^{-w}}{v+w}$ \citep[][Theorem 1]{Hamura:2025}.

The rescaled beta has a right LRV density, such that $f_{\text{RSB}}(x) \approx x^{-1} (\log x)^{-(1+b)}$ as $x\to\infty$. With $0<a<1$, its density diverges at the origin at the polynomial rate of $x^{a-1}$ \citep{Hamura:2025}.
Following Hamura et al.'s suggestion, we set the RSB parameters $a=b=1/2$ in this work.

\section{Mixed Poisson regression under extreme outliers}\label{sec:robustness_extreme_outliers}

\subsection{Robustness of linear regression models to extreme outliers}

As a conceptual background for the later sections, we first summarize the definition of robust linear regression models. 
Robustness against outliers in Bayesian statistics has traditionally been studied in terms of asymptotic outlier rejection.
Based on \cites{Neyman:1971} and \cites{Dawid:1973} work, \citet{OHagan:1979} considered real valued observations $y_1,\dots,y_{n+1}$ and distributions of the form $g(y-\theta)$, with $\theta$ location parameter taking values on the full real line, such as Gaussians, and framed outlier robustness in asymptotic terms, focusing on the posterior distribution of the location parameter when outlying observations approach infinite. 

\begin{definition}{(Asymptotic outlier rejection / Outlier proneness \cite{OHagan:1979})}\label{def:outlierprone}
    Let the observations $y_1, y_2, \dots, y_{n+1}$ be i.i.d. given $\theta$ with densities $p(y_i|\theta)=g(y_i-\theta)$. 
    Then such model is \emph{right (left) outlier prone} of order $n$ if
$$p(\theta|y_1,\dots,y_{n+1})\to p(\theta|y_1,\dots,y_n) \text{ as } y_{n+1}\to +\infty (-\infty).$$

\end{definition}

That is, the effect of a single conflicting observation to the posterior becomes asymptotically negligible as the observation approaches infinity. Right and left outlier-proneness are defined symmetrically, as $y\to +\infty$ and  $y\to -\infty$ respectively. 
As shown by O'Hagan, while Gaussian distributions are not outlier prone, Student's-$t$ distributions, obtained by adding heteroschedasticity to the Gaussian observation model, are outlier prone for the location parameter. 
However, later studies have shown how, in linear regression when both regression coefficients and scale parameter are unknown, Student's-$t$ models are not fully robust, as the convergence defining outlier-proneness is only satisfied for location parameters \citep{Gagnon:2023_st}. 
O'Hagan's definition considers i.i.d. observations, but can easily be extended to linear regression models with densities $p(y_i|\theta)=g(y_i-x_i\beta)$. Note that in linear regression, outliers of type-$y$ (infinite target) and outliers of type-$x$ (infinite covariate) lead to symmetric cases. Indeed, we can identify extreme outliers in terms of infinite error $|y-x\beta|$ \citep{Desgagne:2021}.

\subsection{Defining robustness of  count regression models to extreme outliers}\label{sec:OP_definitions}
Count data models, such as the mixed Poisson models considered here, do not fit into the location-scale family summarized above. Count processes have non-negative support, so we cannot define left outlier proneness similarly to models with support in real line. 
Moreover, count models depend on a positive rate parameter introducing asymmetry between outliers of type-$y$ and type-$x$ as shown below.

Hereafter we consider the mixed Poisson models introduced in Section~\ref{sec:count_models}, and formally define their robustness to large outliers of type-$y$ as well as small and large outliers of type-$x$ (see Table \ref{tab:outlier_types}) in terms of asymptotic outlier rejection. 
We reformulate \cites{OHagan:1979} definition of asymptotic right outlier rejection as follows.
\begin{definition}{(Asymptotic rejection of large outliers of type-$y$)}    \label{def:rightOP_typey_countproc}
A mixed Poisson model performs \emph{asymptotic rejection of large outliers of type-$y$} and order $n$ if
$$p(\beta,\zeta|y_1,\dots,y_{n+1})\to p(\beta,\zeta|y_1,\dots,y_n) \text{ as } y_{n+1}\to + \infty.$$
\end{definition}

The above is equivalent to stating that a mixed Poisson model is \emph{right outlier prone} of order $n$. That is, a count regression model performs asymptotic large outlier rejection of order $n$, if the posterior distribution of the regression model parameters is not affected by one extremely large count. 

Since counts are bounded to non-negative values, we do not attempt to define outlier rejection to small counts in asymptotic terms.
Instead, we study asymptotic rejection of outliers presenting a small count under the concept of large type-$x$ outliers. 
Moreover, it should be noted that the concept of outlier-proneness was first introduced by \cite{Neyman:1971} to refer to the ability of a data generating process to produce extreme observations. It is hence  appropriate for defining asymptotic rejection of outlying targets, but cannot be used to define robustness against outlying covariates, as these are not considered a product of a data generating process. 
However, we can define robustness to outliers with infinite covariates in terms of asymptotic outlier rejection, where an observation $d_{n+1}$ is an extreme outlier of type-$x$ if its covariate vector $x_{n+1}\in\mathbb{R}^p$ is such that for one $j\in \{1,\dots,p\}$ $x_{j,n+1}\to\pm\infty$.
\begin{definition}{(Asymptotic rejection of large (small) outliers of type-$x$)}    \label{def:righlefttOP_typex_countproc}
A mixed Poisson model performs \emph{asymptotic rejection of large outliers of type-$x$} and order $n$ if
$$p(\beta,\zeta|y_1,\dots,y_{n+1})\to p(\beta,\zeta|y_1,\dots,y_n) \text{ as } x_{j,n+1}\to + \infty.$$
Similarly, a mixed Poisson model performs \emph{asymptotic rejection of small outliers of type-$x$} and order $n$ if
$$p(\beta,\zeta|y_1,\dots,y_{n+1})\to p(\beta,\zeta|y_1,\dots,y_n) \text{ as } x_{j,n+1}\to - \infty.$$
\end{definition}

We can extend the above definitions to multiple outliers, by substituting the $n+1$ outlying observation with a set of observations with either infinite or zero target.
Note also that, to avoid notational clutter, we have suppressed the covariates from the conditioning factors of the distributional statements and, as typically done, denoted data only by the target variables $y_1,\dots,y_{n+1}$.

\subsection{Robustness of mixed Poisson regression models to extreme outliers of type-$y$}

Early studies of extreme value theory have shown how the right tails of certain classes of mixed Poisson distributions are asymptotically equal to the right tail of their mixing distribution \citep{Willmot:1990, Perline:1998}.
To study robustness of different mixed Poisson models with GLVs we hence focus on the tail behavior of the mixing component $\eta$. As shown by \citet{Hamura:2025}, having a LRV right tail is a sufficient conditions for such models' robustness to large outliers of type-$y$.

Let us consider the mixed Poisson regression model from Equation~\eqref{eq:Pois-mix-models}, with $y_i$ count response variables for $i=1,\dots,n$ such that $y_i = y_i (\omega) \to \infty$ as $\omega\to\infty$ for $i \in L$ and $y_i$ fixed for $i \in K$, with $K,\ L$ non empty index sets such that $|K| + |L|=n$.
We further assume that $|K^+| + 1 \ge |L| + p$, where $K^+ = \{i \in K | y_i \ge 1\}$ and $p$ is the dimension of $\beta$. 
As proved by \citet[][Theorem 3, reproduced below]{Hamura:2025}, under the assumption that $k+1\ge l+p$, a mixed Poisson asymptotically rejects large type-$y$ outliers if the mixing distribution $\pi$ follows the RSB distribution.

\begin{theorem}{\citep[Theorem 3][]{Hamura:2025}}\label{Th:Pois-RSB_rightOP} 
Suppose that $\pi(\eta_i)$ is a RSB distribution with $k$ degrees of freedom, $\pi(\zeta,\beta) = \pi(\zeta)\pi(\beta)$, $\pi(\zeta)$ is normal, and $\pi(\beta)$ is a proper probability density. Then, we have that $p(\beta, \zeta |y_K,y_L) \to p(\beta, \zeta |y_K )$ as $\omega \to \infty$.
\end{theorem}

The key property making the Poisson-RSB model robust against extremely large counts is the very heavy, LRV, right tail of the RSB mixing distribution.
Note that the Poisson-log-t mixing distribution also has a LRV right tail (Section \ref{sec:count_models}), and is hence robust to extremely large outliers of type-$y$.
We formally present that Poisson-log-t asymptotically rejects large outliers of type-$y$ in the following corollary, and provide its formal proof in the Supplement, Section 1.1 \citep{Pia+Vanhatalo:2026}. 

\begin{corollary}\label{Th:Pois-logt_rightOP} 
Suppose that $\pi(\eta_i)$ is a log-t distribution with $k$ degrees of freedom, $\pi(\zeta,\beta) = \pi(\zeta)\pi(\beta)$, $\pi(\zeta)$ is normal, and $\pi(\beta)$ is a proper probability density. Then, we have that $p(\beta, \zeta |y_K,y_L) \to p(\beta, \zeta |y_K )$ as $\omega \to \infty$.
\end{corollary}

Rejection of small outliers of type-$y$ cannot be defined in asymptotic terms since $y_i$ is bounded by zero from the left. Thus, \cite{Hamura:2025} studied the behavior of the posterior of the Poisson-RSB mixing component at the origin in the presence of zero inflation and asymptotic increase of the (true) rate parameter: 

\begin{theorem}{{\citep[The first part of Theorem 2 of][]{Hamura:2025}}}\label{theom:Hamura_eta}
For any $\epsilon>0$, the posterior probability $Pr(\eta_i <\epsilon|y_i = 0; a, b)$ is a nondecreasing function of $\lambda_i$. 
\end{theorem}

The property that
Pr$({\eta}_{i}<\epsilon | y_{i}=0)$ is a non-decreasing function of $\lambda_{i}$ suggests that the multiplicative error effect captures the zero-valued outlier while the regression component moves away from zero. 
It should be noted that this property does not depend on the mixing distribution and is satisfied by any mixed Poisson, as we show in the Supplement, Section 1.2 \citep{Pia+Vanhatalo:2026}. 
Moreover, and more importantly, this definition is different from the robustness definitions provided in Section~\ref{sec:OP_definitions} and, as will be shown in the next section, does not imply them.

\subsection{Non robustness of mixed Poisson regression models to extreme outliers of type-$x$}\label{sec:non-robustness} 

We now provide the results of the asymptotic behavior of mixed Poisson regression models presented in Section~\ref{sec:count_models} in the presence of outliers of type-$x$.
For simplicity, hereafter we assume $x\in\mathbb{R}$ and a dataset with $n$ regular observation and one outlier of type-$x$. However results can be extended to $\boldsymbol{x}\in\mathbb{R}^p$, where outliers are such that one component $x_j\to\pm\infty$, and to multiple outliers.

When the covariate becomes infinitely large, independently of its direction, mixed Poisson models do not perform asymptotic outlier rejection:

\begin{theorem}\label{Th:nonrobustness_infcovar}
Let $y_{n+1}\in \mathbb{N}$ positive. Under a Bayesian mixed Poisson count regression model framework where the mixing distribution has LRV, or lighter, right tails and polynomial, or lighter, divergence at zero we have that 
$$\frac{p(\beta|y_1,\dots,y_n,y_{n+1})}{p(\beta|y_1,\dots, y_n)}\nrightarrow 1 \text{ as } x_{n+1} \to\pm\infty.$$ 
\end{theorem}

The asymptotic robustness of each mixed Poisson model considered in the presence of outliers of the five alternative cases is summarized in Table \ref{tab:countproc}. 

Furthermore, the asymptotic behavior of the posterior of the regression coefficients changes depending on the target value associated with the outlying covariate.
If the target is positive, the posterior converges in distribution to a point mass at zero:

\begin{theorem}\label{Th:pointmass_post}
Under a mixed Poisson count regression model framework, where the mixing distribution has LRV, or lighter, right tails and polynomial, or lighter, divergence at zero, we have that, for $y_{n+1}\in\mathbb{N}$ positive:
$$p(\beta|y_1,\dots,y_{n+1})\to\delta_0 \text{ as } x_{n+1}\to\pm\infty.$$ 
\end{theorem}

Thus, in contrast to large outliers of type-$y$, in which case the Poisson-RSB and Poisson-log-t provide outlier rejection, the models are not only non-robust to type-$x$ outliers but express very radical ill-behavior through posterior of $\beta$ collapsing to a point mass at zero.
On the contrary, when the target value of the type-$x$ outlier is zero, the posterior conditional to the outlier approaches a truncated posterior given data excluding the outlier:
\begin{theorem}
\label{Th:posterior_zerocounts}
Under a mixed Poisson count regression model framework, we have that, for $y_{n+1}=0$:
\begin{align*}
    p(\beta|y_1,\dots,y_{n}, y_{n+1}=0) &\to  \frac{p(\beta|y_1,\dots,y_n)\mathbb{I}(\beta<0)}{\int_{-\infty}^0 p(\beta|y_1,\dots,y_n)d\beta} \text{ as } x_{n+1}\to\infty,\\
    p(\beta|y_1,\dots,y_{n}, y_{n+1}=0) &\to \frac{p(\beta|y_1,\dots,y_n)\mathbb{I}(\beta>0)}{\int_0^\infty p(\beta|y_1,\dots,y_n)d\beta}\text{ as } x_{n+1}\to-\infty.
\end{align*}

\end{theorem}

The full proofs of the above theorems are available in the Supplement, Section 2 \citep{Pia+Vanhatalo:2026}. 

However, to provide intuition on why the posterior of $\beta$ collapses to a point mass given $y_{n+1}>0$, we provide here a sketch of the proof.
We start by showing that when the target value associated with type-$x$ outlier is positive, the corresponding mixed Poisson likelihood becomes zero everywhere except at zero, indicating a very poor performance of the model. 
We state this formally with the following Lemma:
\begin{lemma}\label{Lemma:likelihood_0}
    Under a mixed Poisson count regression model framework, given $y\in\mathbb{N}$ positive, and $\beta\neq0$, we have that
    $$p_x(y|\beta)=\int_0^\infty \text{Pois}(y|\exp(x\beta)\eta)\pi(\eta)d\eta \to 0 \text{ as }x\to\pm\infty.$$
\end{lemma}

Hence, with extreme outlier of type-$x$, and an associated $y>0$ the origin (in the regression parameter space) becomes the best value of $\beta$ under the mixed Poisson likelihood. 
While Lemma~\ref{Lemma:likelihood_0} holds for any mixed Poisson model with mixing distribution with at most polynomial divergence at zero and at most LRV right tail, for clarity, we show it here for the simple Poisson model: $p_x(y|\beta)= \text{Pois}(y|\exp(x\beta))$.
By applying a change of variable, $\lambda=\exp(x\beta)$, it is easy to see that the likelihood, $\text{Pois}(y|\lambda) = \lambda^ye^{-\lambda}/y!$ goes to zero as $\lambda\to0$ and as $\lambda\to\infty$, given $\beta\neq0$. On the contrary, when $y=0$ the likelihood is non zero only for values of $\beta$ with opposite sign w.r.t. $x$, such that the regression model $\lambda=\exp(x\beta)\to0$. The proof for the mixed Poisson models is given in the Supplement, Section 2 \citep{Pia+Vanhatalo:2026}. 

In addition to Lemma~\ref{Lemma:likelihood_0}, the key fact in the proof of Theorem~\ref{Th:pointmass_post} is to observe that while the marginal likelihood $m_x(y)$ goes to 0 as well, it does that slower than the likelihood $p_x(y|\beta)$ (with $\beta\neq 0$).
By applying the change of variable $z=x\beta$ to the mixed Poisson models we have that
\begin{align*}
    m_x(y) &= \int_{-\infty}^\infty p_x(y|\beta)\pi_\beta(\beta)d\beta \\
    &= \frac{1}{x}\int_{-\infty}^\infty p(y|z)\pi_\beta\Big(\frac{z}{x}\Big)dz\\
    &\approx    \big(\pi_\beta(0)/x\big) \int_{-\infty}^\infty p(y|z)dz \text{ as } x\to\pm\infty
    \end{align*}
where, in the last step, $p(y|z) = \text{Pois}(y|e^z)$ is independent of $x$ and $\pi_\beta$ is bounded by assumption.
This indicates that the marginal decays to 0 at a rate $\frac{1}{x}$, slower than the likelihood.
If we then consider the posterior, by replacing $m(y)$ with its asymptotic equivalent we can show that, when $|\beta|>0$, \citep[see Supplement, Section 2,][]{Pia+Vanhatalo:2026}
$$\lim_{x\to\pm\infty} p_x(\beta|y) = \lim_{x\to\pm\infty} \frac{p_x(y|\beta)\pi_\beta(\beta)}{\pi_\beta(0)/x \int_{-\infty}^\infty p(y|z)dz }= 0.$$
Then, given $a>0$, applying DCT and the property above, we get that:
\begin{align}
\text{Pr}(|\beta|\ge a|y) &= \int_{\{|\beta|\ge a\}} p_x(\beta|y)d\beta \nonumber \\
    &= \frac{ \int_{\{|\beta|\ge a\}} p_x(y|\beta)\pi_\beta(\beta)d\beta}{ \int_\mathbb{R} p_x(y|\beta)\pi_\beta(\beta)d\beta } \nonumber \\
    &\overset{z=x\beta}{=}  \frac{\frac{1}{x} \int_{\{|z|\ge ax\}} p(y|z)\pi_\beta(z/x)dz}{ \frac{1}{x}\int_\mathbb{R} p(y|z)\pi_\beta(z/x)dz }\nonumber \\
    &\overset{\text{DCT}}{\to} \frac{\pi_\beta(0)}{\pi_\beta(0)} \frac{\int_{|z|>\infty} p(y|z)dz}{\int_\mathbb{R} p(y|z)dz}=0 \text{ as } x\to\pm\infty \label{eq:convergence_to_dirac_delta}
\end{align}
that is, the posterior density collapses to the Dirac delta function.

\begin{table}[t]
 \centering
 \caption{Summary of the mixed Poisson and $x$-outlier rejective ($x$-OR) mixed Poisson models, and their mixing distributions, considered in this study and whether they perform (Yes) or not (No) asymptotic rejection of outliers of different types. The data generating process for each mixed Poisson is of the form $y_i\sim \text{Pois}(\lambda_i\eta_i)$ and that of the x-OR mixed Poisson is $y_i\sim \text{$x$-OR Pois}(\lambda_i\eta_i, \theta)$.}

 \begin{tabular}{l l| c|c}
    \toprule
    Distribution & Mixing dist. $\pi$ & \multicolumn{2}{c}{Outlier type}  \\
     & &\ L-$y$ \ & L-$x$ or S-$x$ \\
         
    \midrule
    Poisson &  $\eta_i\sim \delta_1$& No & No   \\
    $x$-OR Poisson & '' & No & Yes  \\ 
    {NegaBin} &  $\eta_i\sim \text{Gamma}(r,r)$ & No & No  \\
    $x$-OR NegBin & '' & No & Yes \\
    {Poisson-log-t} & $\log\eta_i\sim \text{St}(k)$ & Yes & No  \\
    $x$-OR Poisson-log-t & '' & Yes & Yes   \\
    {Poisson-RSB} & $\eta_i\sim \text{RSB}(a,b)$ & Yes & No   \\
    $x$-OR Poisson-RSB & '' & Yes & Yes  \\   
    \bottomrule
    
    \end{tabular} 
     \label{tab:countproc}
 \end{table}

\subsection{Robustness of $x$-OR mixed Poisson regression models to
extreme outliers of type-$x$}\label{sec:robustness-xout}

Theorem~\ref{Th:pointmass_post} shows that the posterior of $\beta$ approaches the Dirac delta function in the presence of type-$x$ outliers (see also Equation \eqref{eq:convergence_to_dirac_delta}). This asymptotic result suggests a method to build a model able to reject these outliers in count regression.
Let us denote the mixed Poisson model likelihood studied so far as:
$$f(y|x'\beta) := \int_0^\infty \text{Pois}(y|\exp({x}'\beta)u)\pi(u)du,$$
with $\beta\in\mathbb{R}^p: \beta\sim \pi_\beta$ (bounded and continuous, often assumed Gaussian)  regression coefficient, and $u>0:u\sim\pi$ mixing distribution.
We now consider the following mixture of mixed Poisson models:
\begin{equation}\label{eq:mixedpois_mix}
p(y|\beta,\theta) = (1-\theta) f(y|x'\beta) + \theta f(y|0),    
\end{equation}
where $\theta\in(0,1)$ is the parameter capturing the proportion of type-$x$ outliers present in the data. We denote this class of models as $x$-OR mixed Poisson regression models, and write $y\sim x\text{-OR Poisson}(\lambda\eta, \theta)$.

Let us consider the same dataset considered in Section \ref{sec:non-robustness}. We can then show that, when $x_{n+1}\to\pm\infty$ this type of models perform asymptotic rejection of type-$x$ outliers, both when $y_{n+1}=0$ and when $y_{n+1}$ is positive.

\begin{theorem}\label{Th:robustness_infcovar}
Let $y_{n+1}\in \mathbb{N}_0$. Under a Bayesian $x$-OR mixed Poisson count regression model framework,
we have that 
$$\frac{p(\beta|y_1,\dots,y_n,y_{n+1},\theta)}{p(\beta|y_1,\dots, y_n,\theta)}\rightarrow 1 \text{ as } x_{n+1} \to\pm\infty.$$ 
\end{theorem}
It should be noted that Theorem \ref{Th:robustness_infcovar} holds for any $x$-OR mixed Poisson, independently of the Poisson mixing distribution tail behavior. 
The behaviors of the studied models under all types of asymptotic outliers are summarized in table~\ref{tab:countproc}.

The $x$-OR mixed Poisson formulation also provides an $x$-outlier detection metric. 
Let us denote by $Z_i\sim \text{Bernoulli}(\theta)$ a latent variable indicating whether a data point is an $x$-outlier or not. 
The conditional posterior probability of this random variable $Z$ is 
\begin{equation}\label{eq:post_bernoulli_xout}
    \Pr(Z_i=1|y_i,\theta,\beta,\eta_i) = \frac{\theta \text{Pois}(\eta_i)}{\theta \text{Pois}(\eta_i) + (1-\theta) \text{Pois}(\exp(x_i'\beta)\eta_i)}.
\end{equation}
We can use the posterior expectation of the above probability as an indicator of how anomalous an observation is in the covariate space. 

\section{Experiments}\label{sec:experiments}

\subsection{Numerical visualization of mixed Poisson model behavior under extreme outliers}

To visualize the results of sections~\ref{sec:OP_definitions}--\ref{sec:non-robustness} we implemented each of the four models in Stan \citep{Stan} and fit them to simulated datasets including 5 regular observations, sampled from a Poisson distribution $\text{Pois}(\exp(a\true+xb\true))$, and one increasingly large outlier of type-$y$ (corresponding to the asymptotic case $y\to\infty$) and decreasingly small outlier of type-$x$ (corresponding to the asymptotic case $x\to-\infty$). 
We first fitted the models to a dataset without an outlier and used the resulting posterior as a baseline.
We then added to the data one outlier at a time, fitted each model to the new data, and compared the resulting posterior with the baseline. 
We repeated this for ten outliers of each type as visualized in figures~\ref{fig:vioppost_yout_extrm}.B and~\ref{fig:vioppost_xout_extrm}.B.

As expected, under the Poisson-RSB and Poisson-log-t models the posterior of the regression coefficient captured the true parameter value without being affected by the outlier of type-$y$ (Figure \ref{fig:vioppost_yout_extrm}.A). In contrast, in the presence of one outlier of type-$x$, the posterior of the regression coefficient gradually concentrated around 0, as the outlier became more extreme, for all models considered (\ref{fig:vioppost_xout_extrm}.A). 
We also visualized the posterior of the mixing component of the outlier observation $\eta\out$ and compared it to $y\out/e^{a\true+b\true x\out}$, in log scale. While for type-$y$ outliers, $\eta\out$ properly captured the outlying observation when assigned a RSB or log-t prior (Figure \ref{fig:vioppost_yout_extrm}.C), it was unable to capture outliers of type-$x$ under any of the models considered (Figure \ref{fig:vioppost_xout_extrm}.C).
Moreover, the posterior of $\eta_{\text{out}}$ did not concentrate at zero as could be assumed by Theorem~\ref{theom:Hamura_eta}.
However, while Theorem~\ref{theom:Hamura_eta} considers the case where non-random $\lambda$ increases, Figure \ref{fig:vioppost_xout_extrm}.C shows the marginal posterior for $\eta_{\text{out}}$ when $\lambda$ is a function of random $\beta$.

\begin{figure}[t]
    \centering
    \includegraphics[width=1\linewidth]{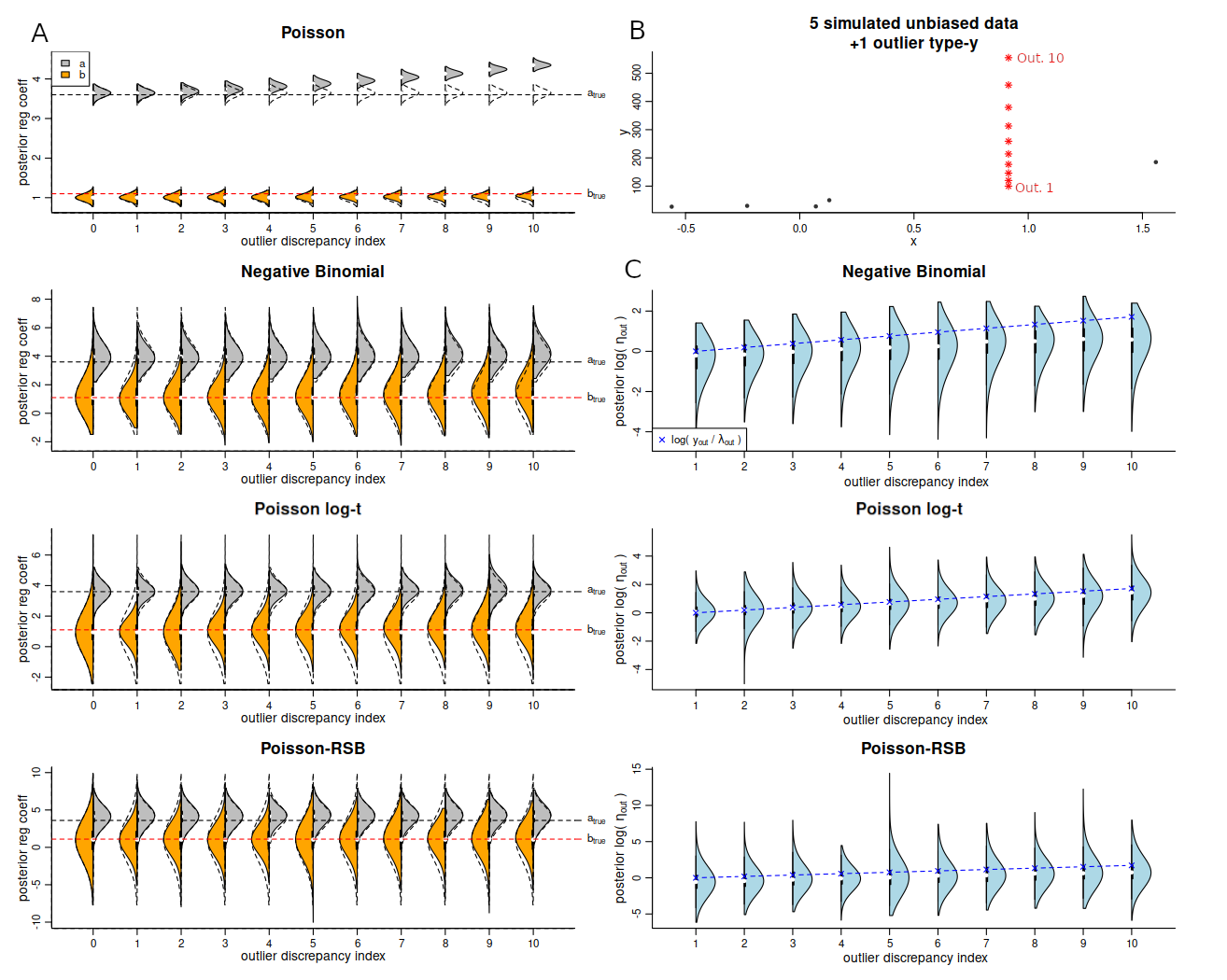}
    \caption{Posterior distributions for intercept ($a$) and regression coefficient ($b$) (A), and multiplicative random effect $\eta\out$ (C), obtained by fitting Poisson and three mixed Poisson models to simulated datasets including five regular observations and one increasingly more extreme outlier of type L-$y$ (B).  Note the differences in y-axis ranges.} 
    \label{fig:vioppost_yout_extrm}
\end{figure}

\begin{figure}[t!]
    \centering
    \includegraphics[width=\linewidth]{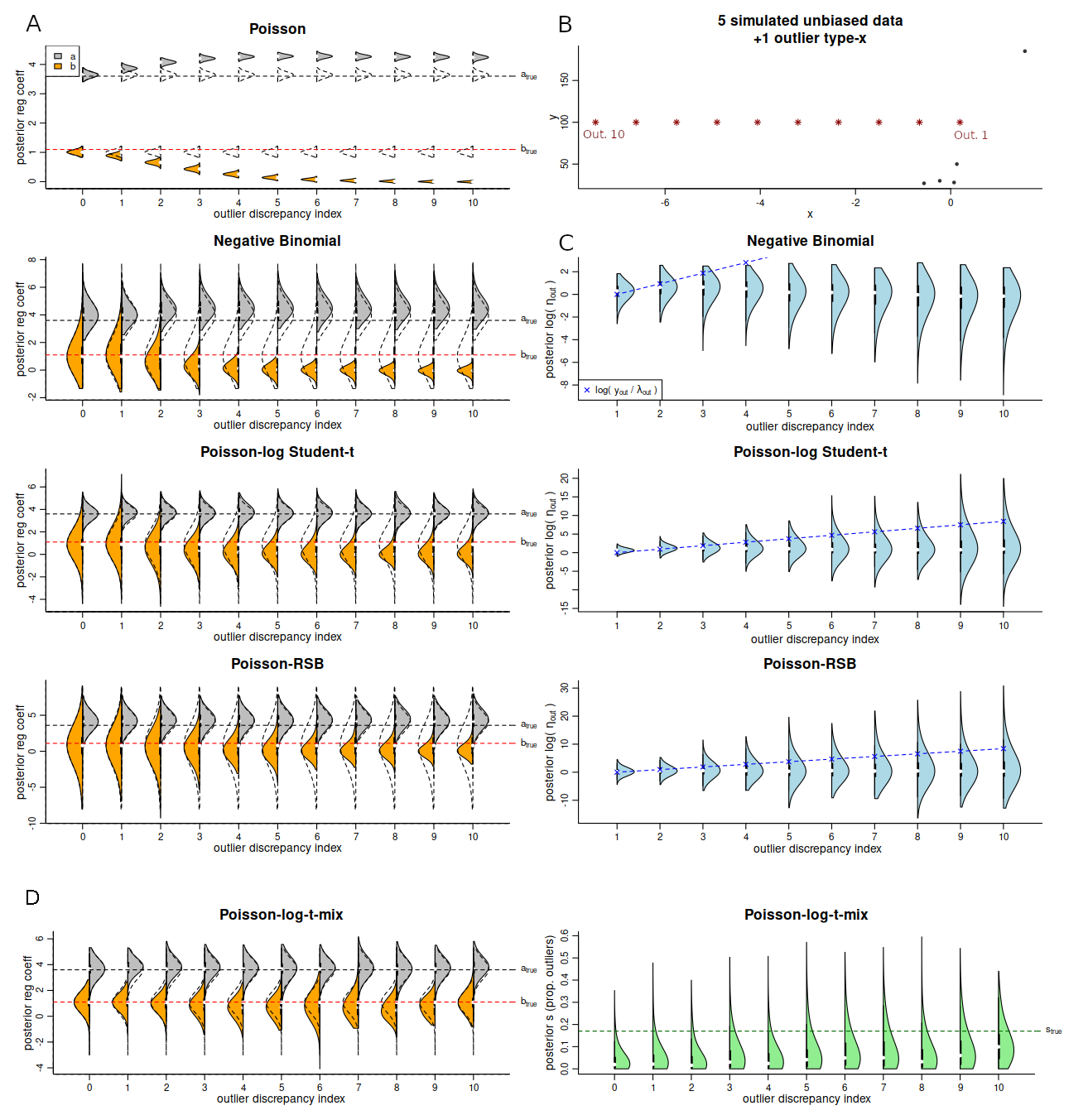}
    \caption{Posterior distributions for the intercept, $a$, and regression coefficient, $b$, (A), and multiplicative random effect $\eta\out$ (C), obtained by fitting Poisson and three mixed Poisson models to simulated datasets including five regular observations and one increasingly more extreme outlier of type S-$x$ (B). Posterior distributions for the intercept, $a$, and regression coefficient, $b$ (gray and orange violinlpots), and proportion of outliers, $s$ (green violinplots), obtained by fitting Poisson-log-t and $x$-OR Poisson-log-t to simulated datasets including five regular observations and one increasingly more extreme outlier of type S-$x$ (D). Note the differences in y-axis ranges.} 
    \label{fig:vioppost_xout_extrm}
\end{figure}


To showcase the robustness results from Section \ref{sec:robustness-xout}, we repeated the experiment presented in the previous section, considering only the Poisson-log-t and $x$-OR Poisson-log-t models, to compare the two formulations. 
While both models showed robustness to outliers of type-$y$, only the $x$-OR model demonstrated robustness to outliers of type-$x$: under this model the regression coefficient posterior does not collapse to a point-mass distribution (Figure ~\ref{fig:vioppost_xout_extrm}.D and Supplement, Section 5 \citep{Pia+Vanhatalo:2026}).
Different choices for the Poisson mixing distribution produced similar results.

\subsection{Behavior of mixed Poisson models under moderate outliers}

While the above theoretical results provide insight on the behavior of the mixed Poisson models under extreme outliers, their properties under moderate outliers can differ.
With moderate (or non-extreme) outliers we mean data points whose component (the target and covariate) values are within the range of the other data, but under some anomalous combinations with respect to the trends in the remaining observations. 

To empirically study how the posterior distribution for the regression coefficients changes in the presence of moderate outlier observations, we set up two simulation experiments. 
Independently of the source of error, outliers can also be classified based on the target value and its expectation conditional on the rest of the data: a data point can be said to be a left hand-side (LHS) outlier if its target value is much smaller than expected, and it is a right hand-side (RHS) outlier if the target value is much higher than expected \citep[following ][]{OHagan:1979}.
We thus considered four different designs for outliers, we distinguish between LHS and RHS outliers and between outliers originated from corrupted target (type-$y$) or covariate (type-$x$). 
We first generated a non-outlier dataset $\mathcal{D}=\{y_i,x_i\}_{i=1}^n$ for two alternative data sizes $n=\{10,100\}$, by generating the covariates $x_i$ from a Gaussian distribution with mean 0 and variance 1 and the counts from a Poisson distribution with log-rate $\log \lambda = a+ bx$.
The parameters of the log rate were sampled from two uniform distributions $b\sim \text{U}(0.001,2.433),\ a\sim \text{U}(4.868-2b, 5.136-2b)$ resulting in data with an expected count ranging from 130 to 170 when $x=2$, and from 0.005 to 150 when $x=-2$. 
Given the non-outlying data, we then generated series of outliers, $\mathcal{D}\out=\{y\out, x\out\}$, with the observed outcome increasingly different from the true Poisson rate, so that the discrepancy $\Delta \log y\out =\log( y\out)-\log(\lambda(x\out))$ increases linearly. We generated 10 outliers of each type as follows:
\begin{itemize}
    \item LHS-$y$: $x\out = \max (x)$ and $y\out \in\{ \lambda(\max (x)), \dots, 0  \}$ decreases to 0
    \item LHS-$x$: $y\out = \lfloor \lambda(\bar{x}) \rfloor$ and $x\out\in\{ \bar x, \dots, 3\max(x) \}$ increases 
    \item RHS-$y$: $x\out = \bar{x}$ and $y\out \in\{  \lfloor \lambda(\bar{x}) \rfloor,\dots, 3\max (y)\}$ increases 
    \item RHS-$x$: $y\out = \lfloor \lambda(\max{x}) \rfloor$ and $x\out \in\{ \bar x, \dots, \min( x\out) \}$ decreases so that $\lambda(\min(x\out))  = \min(y)+10^{-3}$.
\end{itemize}

In the first experiment, we added to the regular dataset one progressively more extreme outlier, leading to 11 different simulated datasets for each outlier type and each data size (see Figure~\ref{fig:result_exp1}.A for an example).
In the second experiment, to test the impact of the number of outliers, we constructed outlier data sets $\mathcal{D}_{\text{out}_j}$ with $j=1,\dots,10$ replicates of the same outlier. We selected five alternative outliers increasingly far from the regular data (n=100), leading to 51 datasets for each outlier type.
We then fitted the four different mixed Poisson models presented in Section~\ref{sec:count_models} and the $x$-OR mixed Poisson presented in Section~\ref{sec:robustness-xout} with log-t mixing distribution, with uninformative Gaussian priors to the linear coefficients $a, b$, to each simulated data (with and without outlier(s)).
We measured each model's robustness in terms of KL divergence (see the next subsection) and repeated the full experimental procedure 20 times to average over the randomness in the data-generating process. 
Note that, we additionally tested the Poisson-log-t model with different degrees of freedom: $k=1,2,4,6$. Since $k=1$ and $k=2$, returned similar KL estimates (see Supplement, Section 5 \citep{Pia+Vanhatalo:2026}), we set $k=1$ in our experiments.

The code necessary to reproduce this simulation study, and the other experiments in the later sections, is available on GitHub\footnote{\texttt{https://github.com/EnvStat/MixedPoissonOutlierRejection}}.

\begin{figure}[!ht]
    \centering
    \includegraphics[width=0.95\linewidth]{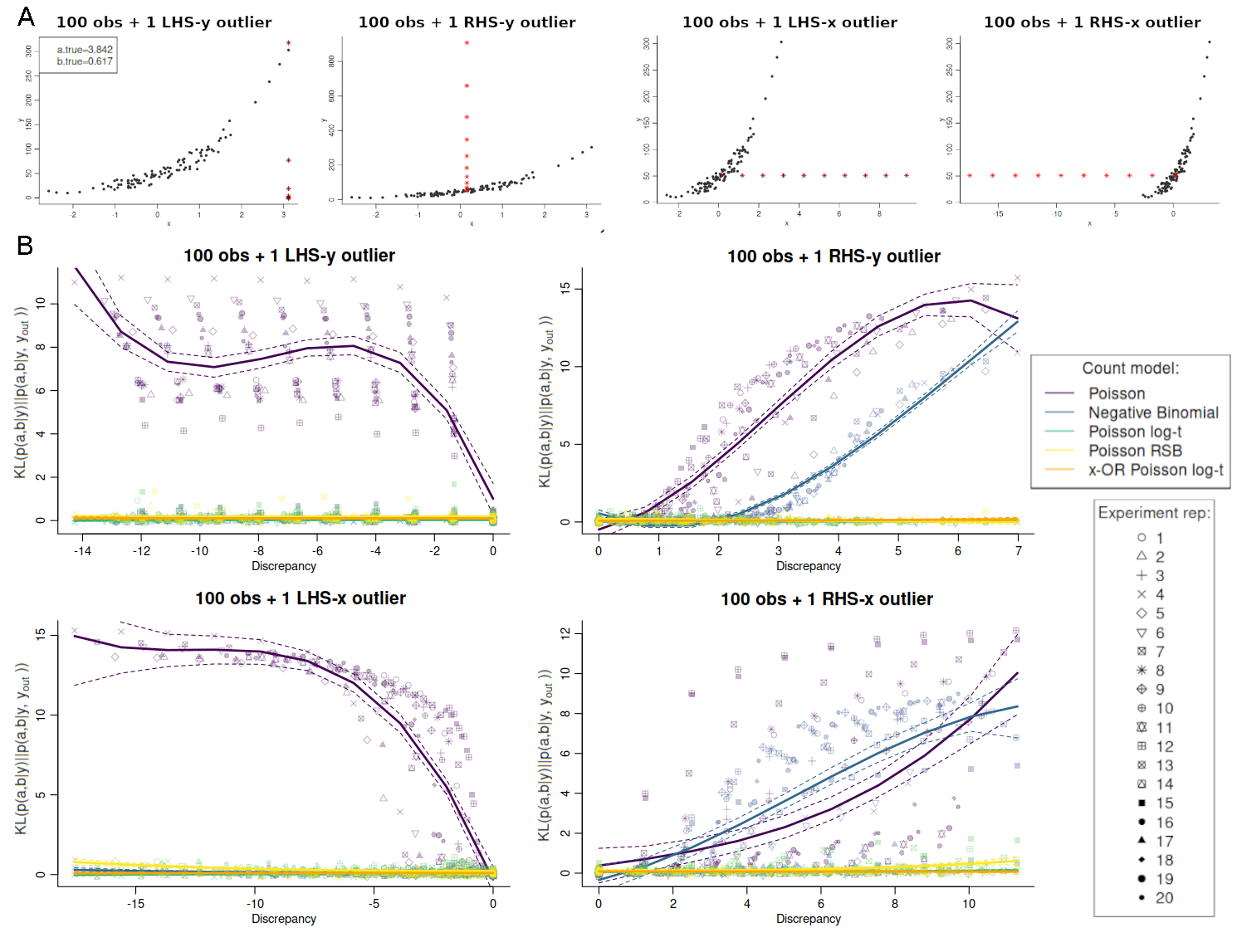} 
    \caption{A) Scatter plots of one simulated dataset: in gray the regular data points, in red the outlier which is gradually getting further from the bulk of the other data in different directions, because of a progressively more corrupted target or covariate value. B) Plots of true discrepancy vs KL divergence between posteriors of the regression coefficients $a,b$ with and without outlier. Results are shown in different plots for the four different types of outlier (panel A), added to 100 regular observations. Different colors correspond to different count models, different point shapes to different simulation repetitions. We fitted a Bayesian polynomial regression of third order to highlight the trend in the KL divergence over different experiment repetitions. The lines represent the expected posterior KL and 95\% CI (dashed lines). 
    }
    \label{fig:result_exp1}
\end{figure}

\begin{figure}[!ht]
    \centering
    \includegraphics[width=\linewidth]{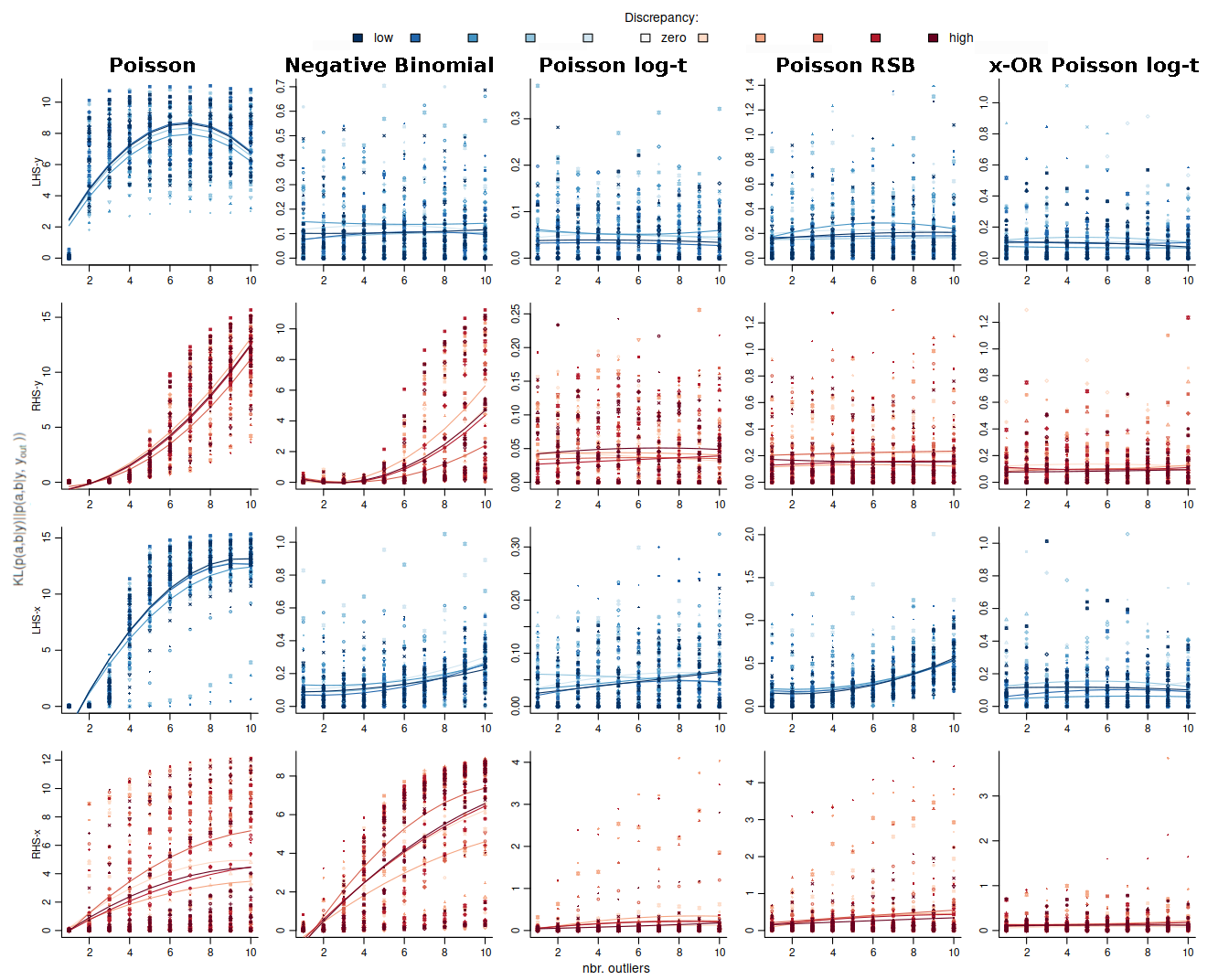}
    \caption{Plots of number of outliers vs KL divergence between posteriors of the regression coefficients $a,b$ with and without outliers (note the different ranges in y-axis). The results are shown in separate plots for the four different types of outlier (rows; see Figure~\ref{fig:result_exp1}.A), and for each count model (columns). We added to 100 regular data points 1 to 10 replicates of the same outlier, and repeated the experiments with outliers increasingly further from the bulk of the data. Different color intensities correspond to different magnitude of the outlier: we used light to dark blue for decreasing discrepancy of LHS outliers, and light to dark red for increasing discrepancy for RHS outliers. Different shapes indicate different simulation repetitions. Trends in the KL divergences highlighted with interpolating lines.
     }
    \label{fig:result_exp_multiouts}
\end{figure}

\subsubsection{Measuring robustness to moderate outliers}\label{sec:robustness_measures}

Assuming the presence of one outlier we measured and compared alternative models' robustness to it via the KL divergence of the model's regression parameters' posterior distribution attained when the outlier was included in the model, with respect to the one obtained after excluding the outlier from the model: 
\begin{equation}\label{eq:kldiv_n+1}
\KL(p(\beta|y_1,\dots,y_n)|| p(\beta|y_1,\dots,y_n, y_{n+1})).    
\end{equation}
The lower the divergence, the more robust the model. The KL-divergence was computed similarly in the presence of multiple outliers.

For all count models considered, we generated Markov chain Monte Carlo (MCMC) samples from the posterior distribution using Stan \citep{Stan}.
%
We then estimated the KL divergence between each compared pair of posterior distributions using the MCMC samples and bias-reduced generalized-NN Divergence Estimator as proposed by \citet{Wang:2009} and available in R in the \emph{kldest} package \citep{kldest24}. 
It should be noted that, when using k-NN density estimates, it is possible to get some negative $\hat{\KL}$ values, indicating that the two posterior distributions are very similar, basically exchangeable. 
Thus, we interpreted negative k-NN estimates as zeros.

\subsubsection{Simulation results}

The estimated trends of KL divergence between posterior with and without an outlier along increasing discrepancy $\Delta y\out$, given 100 regular data points, are shown in Figure~\ref{fig:result_exp1}. 
The Poisson model was strongly affected by all types of moderate outliers. 
The Negative Binomial model was able to handle moderate LHS outliers reasonably well but was sensitive to RHS outliers, originating from either corrupted target or covariate. 
In line with their asymptotic properties, Poisson-RSB and Poisson-log-t were robust against moderate outliers of type $y$ (corrupted target) but Poisson-log-t tended to more completely reject outliers as indicated by lower KL-divergence values. 
Moreover, despite not being asymptotically robust to outliers of type-$x$ (corrupted covariate), these two mixed models were significantly less sensitive to moderate outliers of type-$x$ than Poisson and Negative Binomial.
The $x$-OR Poisson-log-t exhibited robustness to outliers of either type $x$ or $y$.
The experiments with 10 non-outlying observations and one increasingly more extreme outlier returned similar results.

The differences between the models were similar when the number of outliers was increased (Figure \ref{fig:result_exp_multiouts}). The Poisson model was very sensitive to multiple outliers while the Negative Binomial model was mostly sensitive to the increase in number of RHS outliers, and mildly sensitive to LHS outliers of type $y$. The Poisson-log-t seemed to be the least sensitive to all types of outliers while the Poisson-RSB was to some extent sensitive to the increase in the number of outliers of type-$x$.

We further investigated the accuracy of each model's posterior distribution  by measuring the proportion of simulations where posterior samples of $a,b$ included the true parameter value within the 80\% credible interval, under each experimental design. 
In the experiments run with 10 regular observations and one outlier, Poisson-RSB was the best model, achieving almost full coverage, closely followed by Poisson-log-t. 
On the other hand, when experiments were run with 100 regular observations and one outlier, Poisson-log-t performed best, followed by Negative Binomial, while the coverage of Poisson-RSB dropped drastically (see Supplement, Section 5 \citep{Pia+Vanhatalo:2026}).

\section{Kara Sea polar bear case study}\label{sec:kara_sea}

As a motivating real-world case study we tested the alternative count regression models in species distribution modeling of polar bears (\emph{Ursus maritimus}) in the Kara Sea, located north of Siberia.
The data were previously analyzed by \citet{Makinen2018} with the aim of evaluating how the polar bear population is distributed in the area and how its distribution pattern is affected by environmental conditions over the years.
Our target data consist of 9009 polar bear count observations, including one moderate outlier (see below), collected in three different monitoring campaigns in years 1996-2005, 2005--2013, and 2013, respectively. 
The second study campaign consists of presence-only data. 
Following the earlier study, we modeled the observed polar bear counts with respect to distance to the coast, ice concentration, and sea surface salinity, as these covariates were demonstrated to be relevant in explaining polar bears' distribution in the Kara Sea.  
The data were available at a 5 km spatial resolution and at a temporal resolution of one month \citep[for more details see][]{Makinen2018}.

The dataset contained one clear outlier observation: at one sampling location with zero ice three polar bears were counted (the maximum number of observed polar bears in a single site was four in the data), while no polar bears were detected in any other location with ice contraction less than 28\% (see Figure~\ref{fig:respc_karasea}.A).
Also ecologically, the ice concentration is atypical for such an high number of polar bears.
The reason for this outlier was a mistake in the satellite based sea ice concentration measurement (Mäkinen, personal communication).
However, both target and covariate values of this data point are within the range of the remaining data so, according to our previous outlier classifications, this anomalous observation is a moderate outlier of type-$x$.
To assess how much this outlier affects the inference results of the four different count models (Section~\ref{sec:count_models}), we fitted all of them to these data both in the presence and in the absence of the outlier. We then compared how the results differed between these two datasets. 
Moreover, as \cite{Hamura:2025} also introduced a mixing distribution for $\eta$ which allows activating the heavy-tailed RSB mixing distribution only in the presence of extreme observations, we included this version of the Poisson-RSB model in the comparison as well
(see Table \ref{tab:prior_pi_kara_sea} for the prior specification for different choices of the mixing distribution $\pi$).

\begin{table}[t]
\centering
    \caption{Prior specification for the different Poisson-mixing component $\eta_i\sim\pi$ and other parameters, in the Kara Sea case-study.}\label{tabl:kara_sea_models}
    \begin{tabular}{l|l}
    \toprule
    mixed Poisson & Mixing distribution $\pi$ \\
    \midrule
         Poisson& $\eta_i\sim\delta_1$  \\
         Negative Binomial& $\eta_i\sim \text{Gamma}(r,r)$ \\
         & $r\sim\text{Gamma}(1, 0.1)$\\
         Poisson-log-t ($x$-OR) & $\log(\eta_i)\sim \text{St}_2(0,\sigma_\text{st}^2)$; $\theta\sim\text{Beta}(0.5,10)$\\
         & $\log(\sigma_\text{st})\sim  \text{N}(0,1)$\\
         Poisson-RSB & $\eta_i\sim \text{RSB}(0.5,0.5)$\\
         Poisson-RSB $s$-mix &  $\eta_i\sim (1-s)\delta_1 + s\text{RSB}(0.5,0.5)$ \\
         & $s\sim \text{Beta}(1,1)$ \\
         
         \bottomrule
    \end{tabular}
    
    \label{tab:prior_pi_kara_sea}
\end{table}
We considered the following models:
\begin{equation*}
        y_i|\lambda_i,\eta_i \sim \text{Poisson}(\lambda_i\eta_i), \
        \log(\lambda_i) = x_i' \beta + \alpha_{j(i)},\\
\end{equation*}
where we implemented a GLM with the three environmental covariates and an additional fixed effect, $a_j \ j=1,2,3$, accounting for the different sampling methods used in the three monitoring campaigns, and assigned Gaussian priors to all fixed effects.
We further implemented a $x$-OR Poisson-log-t model:
\begin{equation*}
     y_i|\lambda_i,\eta_i,\alpha_{j(i)},\theta \sim (1-\theta)\text{Poisson}(\lambda_i\eta_i) + \theta\text{Poisson}(\exp(\alpha_{j(i)})\eta_i), \
\end{equation*}
and assign to the proportion of outliers parameter a prior $\theta\sim\text{Beta}(0.5,10)$. 
To account for the fact that polar bear counts collected in the campaign 2005--2013 are presence-only data, we used a Poisson distribution truncated at zero for the data points belonging to this sampling campaign.

\subsection{Kara Sea study results}

For each of the models considered, we generated MCMC samples from the posterior distributions of the model parameters using Stan \citep{Stan}, and checked Rhat and effective sample size for convergence.
To assess whether the outlier presence influenced the inference on responses of polar bear abundance to environmental factors, we evaluated the relative change in expected polar bear abundance over the range of each covariate, estimated including and excluding the outlier (Figure \ref{fig:respc_karasea}.A). 
Overall, all count models, with the exception of the Poisson-RSB model, returned similar results: ice concentration had a positive impact on polar bear abundance, while it decreased along increasing salinity and distance from the coast.
The presence of the outlier affected the inferred ice response for all mixed Poisson models considered, so that polar bear abundance was predicted to increase more steeply along increasing ice coverage when the outlier was excluded.
This difference was clearly visible also when comparing the posteriors of the regression coefficients (Figure~\ref{fig:respc_karasea}.B). 
Surprisingly, the Poisson-RSB model, when implemented without the mixture ($s$-mix) formulation, returned posterior estimates with very high variation and failed to properly capture the trends in the data.
In contrast, the presence of a type-$x$ outlier did not have a visible impact on the $x$-OR Poisson-log-t model (Figure~\ref{fig:respc_karasea}).
We also computed the $x$-outlier detection metric, suggested in Equation \ref{eq:post_bernoulli_xout} under the $x$-OR Poisson-log-t model run with all observations. On average, the posterior probability of being an outlier of type-$x$ was significantly higher for the outlier observation compared to the other data points. The expected probability of being a type-$x$ outlier was 0.418 (with (0.0018,0.9855) 80\% CI) for the outlier data point, against an average value of 0.086 (with (0.0022,0.2227) 80\% CI, on average) for the rest of the data points.

\begin{figure} 
    \centering
    \includegraphics[width=.95\linewidth]{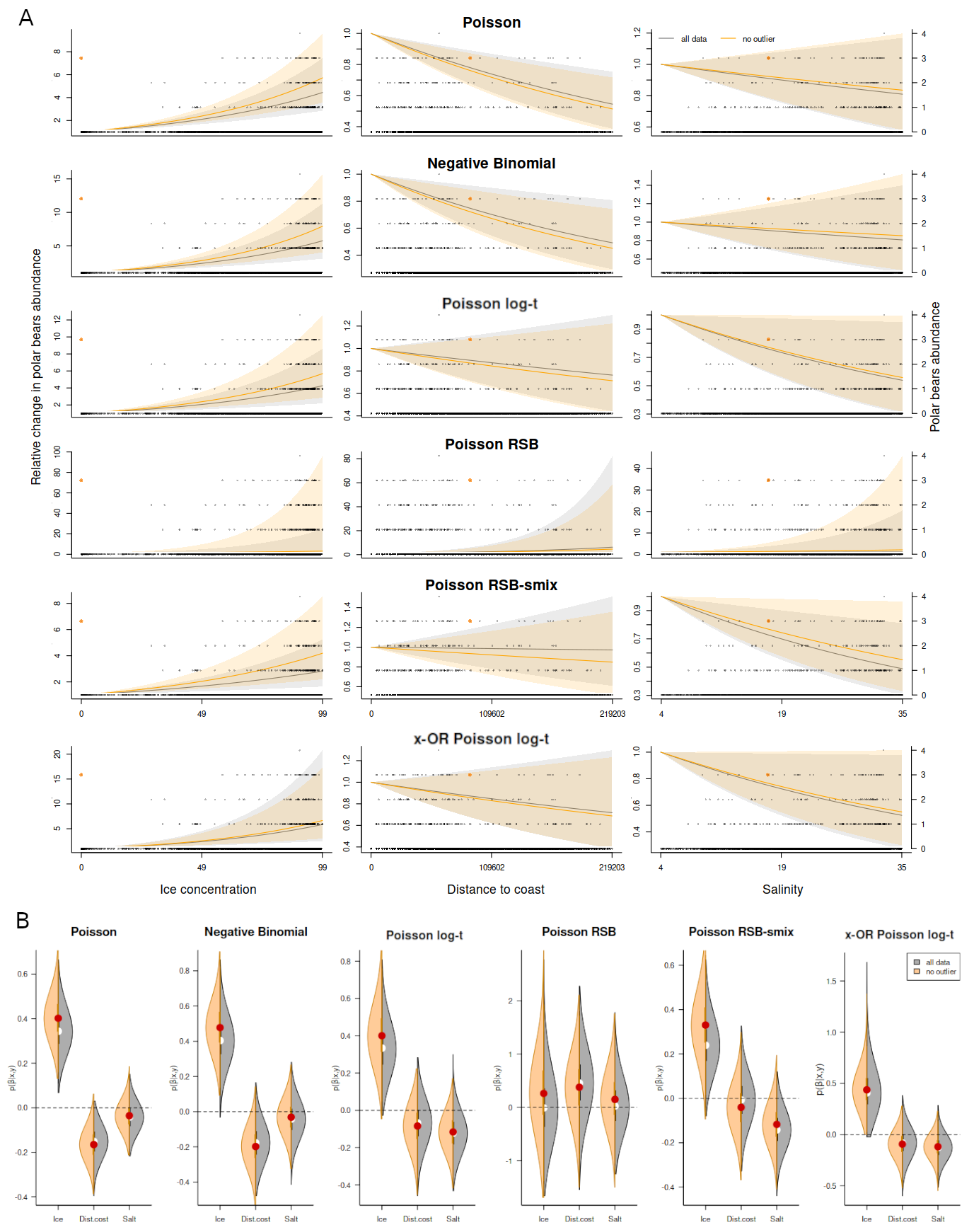}  
    \caption{A) Median and 80\% CI of relative change in polar bear abundance along the range of each covariate, for each count model, both in presence of the outlier observation (gray line) and without it (orange line). The black dots represent the polar bears abundances observed at the corresponding covariate values. The orange dot corresponds to the outlier data point. B) Posterior samples of the regression coefficients associated with the three covariates for each model, both in presence of the outlier observation (gray) and without it (orange).}

    \label{fig:respc_karasea}
\end{figure}

\section{Discussion and conclusion}\label{sec:discussion}

In this work, we extended the definition of robustness of Bayesian models against outliers in terms of asymptotic outlier rejection (following \cites{OHagan:1979} outlier-proneness definition for i.i.d. continuous random variables) to count regression models, and defined asymptotic robustness against outliers originating from both infinite target (type-$y$) and infinite covariate (type-$x$).
We showed how, in mixed Poisson regression, corrupted covariates affect inferential results differently from outliers presenting anomalously large counts. 
Indeed, this contrasts the symmetry between these two types of outliers in linear regression \citep{Gagnon:2020}.
For mixed Poisson models, the mixing component is only able to handle large outlying counts when its tails are sufficiently heavy: a LRV right tail makes the model robust to an infinitely large target, and a left tail with polynomial diverging behavior allows the model to capture zero inflated data \citep{Hamura:2025}.
However, as shown here (Theorem~\ref{Th:nonrobustness_infcovar}), independently of the mixing component tails, mixed Poisson models are affected by outliers originating from corrupted covariates. 
Asymptotically, as the outlying covariate becomes infinitely large (or small), the likelihood shrinks to zero for any value of $\beta$ in the regression parameter space, with the exception of the case $\beta=0$. As a result, the posterior collapses to a point mass distribution, concentrating at zero. 
This phenomenon shares similarity with the Lindley's paradox \citep{Lindley:1957} in the sense that a very poor model fit indicated by a zero likelihood (almost everywhere) leads to an extremely precise regression coefficient posterior.

At a close look, this degenerate behavior makes perfect sense though. In the presence of an infinite covariate, the best posterior estimate concentrates around zero, as this is the only way to cancel the effect of such extreme outlier.
In contrast, the $x$-OR mixed Poisson models presented in this work are able to accommodate outliers of type-$x$. The mixed Poisson mixture component with rate independent of covariates is able to capture anomalous covariate observations, so that the regression coefficient posterior is not affected by them.

As demonstrated by our simulation studies and real-world case study, mixed Poisson models' lack of asymptotic robustness against outliers of type-$x$, did not immediately induce poor model performance in the presence of moderate, finite, outliers. Often, when the outlier was still within the range of the rest of the data, already a Negative Binomial (Poisson-Gamma) model was sufficient to properly capture the structure of the studied data even though parameter estimates were biased. 
Interestingly, the Negative Binomial model also had a better fit to the data in many of our moderate outlier of type-$x$ experiments than the Poisson-RSB model. 
Specifically, the extreme tail of the RSB distribution let the variance in the posterior estimates explode if applied to all data points, indicating that this model might not be appropriate to fit data that include only moderate outliers.
This behavior was fixed when defining the mixing prior as a mixture of RSB and point mass at one (the Poisson-RSB $s$-mix model), as shown in the Kara Sea case study.

However, our case study on polar bear distribution in the Kara Sea demonstrates that a very small proportion of moderate outliers (here $1/9009 < 0.01\%$) can already have a significant effect on predictive results. 
While the changes in the inferred coefficient values for ice coverage with and without the outlier seem marginal (Figure~\ref{fig:respc_karasea}.B), they are large from a predictive point of view. 
Under all mixed Poisson models, without the outlier, polar bear density was predicted to increase approximately six fold from the lowest ice concentration to the highest ice concentration, whereas when the outlier was included in the analysis, the increase was only four fold (Figure~\ref{fig:respc_karasea}.A). 
A one third decrease in the predicted relative change can have substantial implications, for example, in ecological risk assessment \citep{Helle+etal:2020}.
The result thus highlights the importance of future work towards count models that are robust for outliers in covariates -- as well as methods to detect them \citep{Plummer:2025}.

The theory and methods for robust Bayesian modeling have taken significant steps during the recent years \citep{OHagan+Pericchi:2012,Desgagne:2015,Gagnon:2020, Gagnon:2023, Hamura:2025}.
However, we still lack a comprehensive understanding of the robustness properties of Bayesian regression models. 
Specifically, even though the duality of a data-point in regression setting has long been acknowledged in robust regression, its impact  on outlier rejection in generalized linear regression has been left for lesser attention. 
To our knowledge, this is the first work that formally studied infinite covariate outliers in GLMs. 
Since other GLMs share similarity to count regression models in that their likelihoods are not symmetric to extreme target and covariate values, our work also has important implications beyond count regression.
For example, \cite{Gagnon:2023} discussed extreme outliers in covariates in the context of Gamma regression but provided formal robustness results only for extreme target values. 

Specifically, our work provides several important insights into Bayesian count regression modeling.
First, our formal proof that mixed Poisson models do not behave symmetrically with respect to outliers of type-$y$ and type-$x$ has important practical implication. 
While practitioners can use LRV heavy tailed models (e.g., Poisson-log-t or Poisson-RSB $s$-mix) to secure against corrupted target values, they should still check their model results for high leverage and influence points to identify potential covariate outliers \citep{Chatterjee:1986,Plummer:2025}.
Second, our proposed method consisting of introducing a spike-and-slab type of mixture formulation for data point-wise regression coefficients is shown to grantee asymptotic rejection of outliers in the covariates.
The results presented here also provide insight into future methods development by indicating that introducing a mixing component is not enough to make count regression models asymptotically robust against outliers of type-$x$. 
After all, the mixing component $\eta$ is part of the data generating process by construction, while in GLMs covariates are considered fixed by design. 
Third, our Corollary~\ref{Th:Pois-logt_rightOP} shows that the log Student's-$t$ mixed Poisson model is robust to large outliers of type-$y$. 
While this result was a straightforward application of the theoretical work of \citet{Hamura:2025}, together with our empirical results, it provides useful insights to practitioners. Our simulation and case study results indicate that the log Student's-$t$ mixing distribution might behave more nicely in the presence of moderate outliers than the heavier tail rescaled Beta distribution.
Fourth, our results also motivate further studies on how different types of regression models, such as the logistic, behave in the presence of corrupted covariates. 
Additionally, it could be worth studying what happens in the presence of outliers presenting corrupted values in multiple directions.

We also note that a rigorous and theoretical definition of moderate outliers and of the corresponding model robustness against them is still lacking in Bayesian statistics. 
Developing a cohesive, non-asymptotic theory of robustness would be a valuable contribution to the field, especially useful to applied statisticians.
However, these questions were beyond the scope of this work.

\section*{Acknowledgement}
We thank Shu Zhen Tan for her help in initial tests and experiments at the beginning of this study and express our gratitude to Petteri Piiroinen and Eetu Halme for their help in proofreading the proofs of the theorems presented in this article and for discussions that helped improve them.
We thank Dario Gasbarra for helpful discussions related to the theoretical part of the work.
The authors would also like to thank the anonymous referees, an Associate Editor and the Editor for their constructive comments that improved the quality of this paper.
This work has received funding from the Academy of Finland (grants 317255 and 1368100) and the European Union (ERC Consolidator Grant BEFPREDICT, 101087409). In addition, IP acknowledges funding from the Doctoral Programme in Mathematics and Statistics of the University of Helsinki.

\newpage

\appendix
\section{ }\label{appendix}
\setcounter{equation}{0}
\setcounter{figure}{0}
\setcounter{table}{0}
 \renewcommand{\theequation}{\thesection.\arabic{equation}}
 \renewcommand{\thefigure}{\thesection.\arabic{figure}}
 \renewcommand{\thetable}{\thesection.\arabic{table}}

\subsection{Mixed Poisson models' robustness to outliers of type-$y$: theoretical results}


\subsubsection{Proof of Corollary 3.5: Poisson-log-t robustness to infinite counts}\label{appendix:robustness_Poislogt}

Corollary 3.5 ensures theoretical robustness to large outliers of type-$y$ for the mixed Poisson-log-t model, and is a direct consequence of Theorem 3 of \citet{Hamura:2025}.

\emph{proof:}

Let us denote by $f_\pi(y;z) = \int_0^\infty \text{Pois}(y|e^z u)\pi(u)du$.

We have assumed that $\pi(u)$ follows a log-t distribution, with $k>0$ degrees of freedom.
For $f_{\log-t}(y;z)$ we hence have that $\pi(u)\sim \frac{1}{u}\frac{1}{(\log(u))^{1+k}}$ as $u\to\infty$, meaning the log-t distribution is right LRV, with index $k+1$.

Following Lemma S1 from \citet{Hamura:2025}, if the mixing distribution $\pi$ is right LRV, then, $p(y|\beta,\zeta)/ f_\pi(y;0)\to 1$ as $y \to+\infty$.

We denoted with $y_K$ the set of non-outlying observations and with $y_L$ the set of infinitely large outliers of type-$y$, such that for $i\in L$, $y_i = y_i(\omega)\to\infty$ as $\omega\to\infty$.
Applying Bayes Theorem, conditional independence of observations, and writing marginals as integrals w.r.t. $\pi(\beta,\zeta)$ we get:
\begin{equation*}
    \begin{split}
        \frac{p(\beta,\zeta|y_K, y_L)}{p(\beta,\zeta|y_K)} &= \frac{p(y_K|\beta,\zeta)p(y_L|\beta,\zeta)\pi(\beta,\zeta)}{p(y_K|\beta,\zeta)\pi(\beta,\zeta)} \frac{p(y_K)}{p(y_K,y_L)}\\
        &= \frac{p(y_L|\beta,\zeta)\int_{\mathbb{R}^p\times\mathbb{R}^q} p(y_K|\beta,\zeta)\pi(\beta,\zeta)d\beta d\zeta}{\int_{\mathbb{R}^p\times\mathbb{R}^q} p(y_L|\beta,\zeta)p(y_K|\beta,\zeta)\pi(\beta,\zeta)d\beta d\zeta} \\
        &= \frac{p(y_L|\beta,\zeta)\int_{\mathbb{R}^p\times\mathbb{R}^q} p(y_k|\beta,\zeta)\pi(\beta,\zeta)d\beta d\zeta}{\int_{\mathbb{R}^p\times\mathbb{R}^q} p(y_L|\beta,\zeta)p(y_K|\beta,\zeta)\pi(\beta,\zeta)d\beta d\zeta} \frac{f_\pi(y_L;0)}{f_\pi(y_L;0)}\\
        &\overset{\text{DCT+L}}{\longrightarrow} \frac{\int_{\mathbb{R}^p\times\mathbb{R}^q} p(y_K|\beta,\zeta)\pi(\beta,\zeta)d\beta d\zeta}{\int_{\mathbb{R}^p\times\mathbb{R}^q} p(y_K|\beta,\zeta)\pi(\beta,\zeta)d\beta d\zeta} = 1 \text{ as } \omega\to +\infty
    \end{split}
\end{equation*}
where in the last step, we used the Dominated convergence Theorem (DCT) and Lemma S1 from \citet{Hamura:2025}.
The hypotheses of normality and boundedness on $\pi(\zeta)$ and $\pi(\beta)$ respectively, together with the assumption that $|K^+| + 1 \ge |L| + p$ are sufficient to satisfy DCT conditions, as proved in Lemma S4 from \citet{Hamura:2025}.
$\square$

\subsubsection{Concentration of the probability mass of the mixing component at zero with zero counts and increasing intensity}\label{appendix:robustness_zeros}

\begin{theorem}
Given a mixed Poisson model: $y_i\sim \text{Poisson}(\lambda_i\eta_i)$, $\eta_i\sim \pi$ noise effect, and $\lambda_i$ regression model, for any $\epsilon>0$, the posterior probability $P(\eta_i<\epsilon|y_i=0)$ is a non decreasing function of $\lambda_i$.

\begin{proof}
    For any $\lambda,\Delta>0$, we have that:

    \begin{equation*}
        \begin{split}
            \left\{ \Pr(\eta_i<\epsilon|y_i=0) \right\}\big|_{\lambda_i= \lambda} &= \frac{\int_0^\infty\chi_{(0,\epsilon)}(u) e^{-\lambda u}\pi(u)du}{\int_0^\infty e^{-\lambda u}\pi(u)du}\\
            &\overset{*}{\le} \frac{\int_0^\infty\chi_{(0,\epsilon)}(u) e^{-\lambda u}e^{-\Delta u}\pi(u)du / \int_0^\infty e^{-\lambda u}\pi(u)du }{ \int_0^\infty e^{-\lambda u}e^{-\Delta u}\pi(u)du / \int_0^\infty e^{-\lambda u}\pi(u)du} \\
            &= \left\{ \Pr(\eta_i<\epsilon|y_i=0) \right\}\big|_{\lambda_i= \lambda+\Delta} 
        \end{split}
    \end{equation*}
    
    In $*$ we used the covariance inequality for:
    $$\E_\rho[\chi_{(0,\epsilon)}(u)]\le \frac{\E_\rho[e^{-\Delta u}\chi_{(0,\epsilon)}(u)]}{\E_\rho[e^{-\Delta u}]}$$
    where $\chi_{(0,\epsilon)}(u)$ and $e^{-\Delta u}$ are non increasing function in $u\in(0,\infty)$ and $\rho(u) = \frac{e^{-\lambda u}\pi(u)}{\int_0^\infty e^{-\lambda u}\pi(u) du}$ is the conditional distribution for $p(u|y=0)$.

    Therefore, $\Pr(\eta_i<\epsilon|y_i=0)$ is a non decreasing function of $\lambda_i$. 
\end{proof}
\end{theorem}

\subsection{Mixed Poisson models non-robustness to outliers of type-$x$: theoretical results}\label{appendix:nonrob_infcov}

Let us consider a set of $n$ observations, $\mathcal{D}_n = \{ d_i : d_i = (y_i,\boldsymbol{x}_i) \in\mathbb{N}\times\mathbb{R}^p,\  i=,\dots,n\}$, modeled with a mixed Poisson count regression with Gaussian latent variables, under a Bayesian framework:
\begin{align*}
    y_i|u_i,\beta&\sim \text{Pois}(\exp(\boldsymbol{x}_i'\beta)u_i)\\
    \beta&\sim \text{N}(\cdot)\\
    u_i&\sim\pi : \ u_i>0
\end{align*}
where $\beta$ is a $p$-dimensional vector of Gaussian regression coefficients 
and $u_i$ are the multiplicative error components capturing outlier observations following a continuous mixing distribution $\pi$. Specifically, we are interested in mixing distributions divergent at 0, and with a very heavy log-regularly varying (LRV) right tail, as these distributions have been proven to be robust to infinitely large counts and able to handle zero-inflation \citep{Hamura:2025}.
That is, we assume that the mixing distribution tails are such that 
\begin{align}
    \pi(u)&\approx u^{a-1} \text{ as } u\to0, \ 0<a<1, \label{eq:right_tail_ap}\\
    \pi(u)&\approx \frac{1}{u\log(u)^{b+1}} \text{ as } u\to\infty, \ b>0. \label{eq:left_tail_ap}
\end{align}
We further assume the $n$ observations to be conditionally independent given $\beta$. 
We define an observation $i$ as outlier if there is an index $j\in 1,\dots,p$ such that $|x_{i,j}|\to\infty$.

For simplicity, hereafter we assume $x\in\mathbb{R}$. 
We also consider a dataset with two observations only, $(y,x), (y\out,x\out) \in\mathbb{N}\times\mathbb{R}$ (one outlier and one not). The result can then be extended to multiple observations.
Note also that the following results can be obtained for mixed Poisson regression models whose regression coefficients $\beta$ follow any continuous bounded prior $\beta\sim\pi_\beta$.

To prove that Bayesian mixed Poisson count regression models are not robust to outlying covariates, in a Bayesian sense, we want to show that as the outlier covariate gets infinitely large, the model does not perform asymptotic outlier rejection, that is:

\begin{equation}\label{eq:outlier_nonrob}
\frac{p(\beta|y, y\out)}{p(\beta|y)} \nrightarrow 1 \text{ as } x\out\to\pm\infty.    
\end{equation}

We will first show that, given $y_{\text{out}}>0$, the posterior distribution of the regression coefficient $\beta$ given an outlier with infinitely large covariate degenerates to the point mass distribution at 0, as $x\out\to\pm\infty$, that is $p(\beta|y\out)\to \delta_0 \text{ as } x\to\pm\infty.$
We will hence derive the result in Equation \ref{eq:outlier_nonrob}.

\vspace{1cm}

Let us denote by $p_x(\beta|y)=p_x(y|\beta)\pi_\beta(\beta)/m(y)$ the posterior, where $m(y) = \int_\mathbb{R}p_x(y|\beta)\pi_\beta(\beta)d\beta$ is the marginal, $p_x(y|\beta) = \int_0^\infty\text{Pois}(y|e^{x\beta}u)\pi(u)du$ is the mixed Poisson likelihood, and $\pi_\beta(\beta)$ is a continuous bounded prior. 

To prove that $p_x(\beta|y)$ converges in the sense of distribution to $\delta_0$ we will first check that it satisfies the following properties and then show that satisfying such property is a sufficient condition for $p_x(\beta|y)$ to converge to $\delta_0$ \citep[][ properties 5.5.1]{choksi2022}.

\begin{lemma}\label{Lemma:likelihood_0_ap}
    Under a mixed Poisson count regression model framework, given $y\in\mathbb{N}$ positive, and $\beta\neq0$, we have that
    $$p_x(y|\beta) \to 0 \text{ as }x\to\pm\infty.$$
\end{lemma}

\emph{proof:}

Under mixed Poisson model assumptions, the likelihood is such that:
\begin{align*}
    p_x(y|\beta) &= \int_0^\infty \text{Pois}(y|e^{x\beta}u)\pi(u)du\\
    &= \int_0^\infty u^{y}\pi(u)/y! \exp(y x\beta - u e^{x\beta}) du
\end{align*}

We are interested in studying the limit for $x\to \pm \infty$, given $|\beta|\ge 0$. 
As $x$ in the likelihood component always appears multiplied by $\beta$, the limit cases we are interested in are $x\beta \to \pm \infty$, given $\beta\neq 0$.
Let us denote $w = e^{x\beta}$, then
$$\lim_{w\to\infty}  \text{Pois}(y|wu) = \lim_{w\to\infty} (wu)^ye^{-wu}/y! = 0$$
$$\lim_{w\to0}  \text{Pois}(y|wu) = \lim_{w\to0} (wu)^ye^{-wu}/y! = 0.$$

Hence, applying the Dominated convergence Theorem (DCT), we get that
$$\lim_{x\beta\to\pm \infty} p_x(y|\beta) \overset{DCT}{=}  \int_0^\infty u^{y}\pi(u)/y!\Big\{ \lim_{x\beta\to\pm \infty} \Big(\exp(y x\beta - u e^{x\beta})\Big) \Big\} du = 0.$$

Note that in order to apply the DCT, and swap limit (w.r.t. $x$) and integral (w.r.t. $u$), we need to check that the DCT hypotheses are satisfied.
That is, we need to check that $\lim_{x\to\pm\infty} \text{Pois}(y|e^{x\beta}u)\pi(u) =f(u)$ exists for all $u>0$ (which we just checked), and that for each $x\in\mathbb{R}$ there exists an integrable function independent of $x$, such that $\text{Pois}(y|e^{x\beta}u)\pi(u)\le g(u)\ \forall u>0$. 

Given $y,w = e^{x\beta}>0$, then: 
\begin{align*}
     \text{Pois}(y|wu)\pi(u) &= w^y u^y e^{-wu}/y!\pi(u) \\
     &\le  y^y e^{-y} /y! \pi(u)
\end{align*}
since the Poisson pmf $\text{Pois}(y|wu)$ reaches its maximum for $u=y/w$, hence $g(u) = \text{Pois}(y|y)\pi(u)\ge \text{Pois}(y|wu)\pi(u)$ for all $u,w>0$. Further, $g(u)$ is integrable since $\int_0^\infty \text{Pois}(y|y)\pi(u)du = \text{Pois}(y|y)<\infty$.
$\square$

\begin{lemma}\label{Lemma:xlikelihood_0}
    Under a mixed Poisson count regression model framework,  where the mixing distribution has LRV, or lighter, right tails and polynomial, or lighter, divergence at zero, given $y\in\mathbb{N}$ positive, and $\beta\neq0$, we have that
    $$xp_x(y|\beta) \to 0 \text{ as }x\to\pm\infty.$$    
\end{lemma}

\emph{proof:}
Let us split the integral in the two intervals $(0,1)$ and $(0,\infty)$. We have that
$$xp_x(y|\beta) = \int_0^1 \frac{\pi(u)u^y}{y!}x \exp(x\beta y - e^{x\beta}u )du + \int_1^\infty \frac{\pi(u)u^y}{y!}x \exp(x\beta y - e^{x\beta}u )du.$$
We now study the limiting behavior of both integrals, as $x\to\pm\infty$.

Let us first consider the integral in $(1,\infty)$. 
Given $y\in\mathcal{N}$ positive, we have that $x\beta e^{yx\beta}<e^{(y+1)x\beta}$. For a fixed $u>1$, we can bound the integrand, $f(x)$, as follows:
$$f(x) = \frac{\pi(u)u^y}{y!}x \exp(x\beta y - e^{x\beta}u )< \frac{\pi(u)u^y}{y!\beta}\exp((y+1)x\beta -e^{x\beta}u ) := h(x).$$
Then 
\begin{align*}
    \frac{dh(x)}{dx} \propto (y+1)\beta -\beta u e^{x\beta}=0 \iff \hat{x}= \frac{1}{\beta}\log\left(\frac{y+1}{u}\right)
\end{align*}
By evaluating $h$ at $\hat x$ we hence get that
\begin{align*}
\max_{x\in\mathbb{R}} h(x) & = h(\hat{x})\\
&= \frac{\pi(u)u^y}{y!\beta} \exp((y+1)\log\frac{y+1}{u} - (y+1))\\
& = \frac{(y+1)^{y+1}}{y!\beta} e^{-(y+1)} \frac{\pi(u)}{u} := g(u)    
\end{align*}
and $g(u)$ is integrable in $(1,\infty)$, as $\int_1^\infty g(u)du\propto \int_1^\infty \frac{\pi(u)}{u}du$, where the integrand evaluated at the limits of the integral is such that $\pi(1)$ is finite, and as $u\to\infty$, $\frac{\pi(u)}{u}\approx\frac{1}{u^2(\log u)^{1+b}}\to 0$ faster than $\frac{1}{u}$, when $\pi$ is right LRV (Equation \ref{eq:right_tail_ap}), with index $b$. Note that for any proper integrable density $\pi$ with right tail LRV or lighter, $\frac{\pi(u)}{u}\to 0$ faster than $\frac{1}{u}$.
We therefore have $\int_1^\infty g(u)du<\infty$ and $\frac{\pi(u)u^y}{y!}x \exp(x\beta y - e^{x\beta}u )<g(u)$.
We can hence apply DCT to study the integral in $(1,\infty)$:
$$\lim_{x\to\pm\infty}\int_1^\infty \frac{\pi(u)u^y}{y!}x \exp(x\beta y - e^{x\beta}u )du \overset{DCT}{=} \int_1^\infty \lim_{x\to\pm\infty} \frac{\pi(u)u^y}{y!}x \exp(x\beta y - e^{x\beta}u )du=0.$$

Let us now consider the integral in $(0,1)$: 
$\int_0^1 \frac{u^y\pi(u)}{\beta y!}w^y\exp(-wu)\log (w) du$
where we denote by $w = e^{x\beta}$.

For the limit $w\to 0$ ($x\beta\to-\infty$), we use that $e^{-wu}<1$ and get that
$$\int_0^1 \frac{u^y\pi(u)}{\beta y!}w^y\exp(-wu)\log w du \le w^y\log(w) \int_0^1 \frac{u^y\pi(u)}{\beta y!}du\to 0 \text{ as } w\to0$$
since, the integral is finite as when $\pi$ left tail diverges at zero with rate $a-1$ (Equation \ref{eq:left_tail_ap}), $u^y\pi(u)\approx u^{y+a-1}\to 0$ as $u\to0$, for any $y>0$ and $a\in(0,1)$.

Let us now consider the limit $w\to\infty$ ($x\beta\to\infty$).
By applying the change of variable $v = uw$, we have that
\begin{align*}
\int_0^1 \frac{u^y\pi(u)}{\beta y!}w^y\exp(-wu)\log (w) du
&= \int_0^w \frac{v^ye^{-v}}{\beta y!}\log(w) \frac{\pi(v/w)}{w} dv\\
&\le \frac{\log(w)}{w} \int_0^w \sup_{v\in(0,w)}\{ \pi(v/w)\}\frac{e^{-v}v^y}{\beta y!} dv \\
&\approx \frac{\log(w) }{w} \int_0^w \frac{v^{a-1}}{w^{a-1}} \frac{e^{-v}v^y}{\beta y!} dv\\
&= \frac{\log(w) }{w^a} \int_0^w \frac{e^{-v}v^{y+a-1}}{\beta y!} dv\\
& \to 0 \text{ as } w\to\infty, y>0.
\end{align*}
where we used that the supremum of $\pi$ is reached around 0, so that $\pi(w/v)\approx (w/v)^{a-1}$. Further, the integral is finite as, it is proportional to the integral of a Gamma random variable over $(0,w)$: $\int_0^w \frac{v^{y+a-1}e^{-v}}{ y!} dv =\frac{(y+a)!}{y!} \Pr(V\le w)\to \frac{(y+a)!}{y!}$ as $w\to\infty$, with $V\sim\text{Gamma}(y+a,1)$.

Thus, we have just shown that both integrals in $(0,1)$ and $(1,\infty)$ go to 0, meaning $xp_x(y|\beta)\to0$ as $x\to\pm\infty$.
$\square$

\begin{lemma}\label{Th:pointmass_post_suffconds}

Under a mixed Poisson count regression model framework, where the mixing distribution has LRV, or lighter, right tails and polynomial, or lighter, divergence at zero, given $y\in\mathbb{N}$ positive, we have that
$p_x(\beta|y)$ satisfies the three following properties:

\begin{enumerate}
    \item Non negativity: $p_x(\beta|y)\ge0 \ \forall \beta\in\mathbb{R}, x\in\mathbb{R}$.
    \item Unit mass: $\int_{-\infty}^\infty p_x(\beta|y) d\beta = 1$.
    \item Concentration at zero: $\forall a>0 \ \lim_{x\to\pm\infty} \sup_{|\beta|\ge a} |p_x(\beta|y)|=0$ and 
    \item[ ] \hspace{3.5cm} $\forall a>0 \ \lim_{x\to\pm\infty} \int_{\{|\beta|\ge a\}} p_x(\beta|y)d\beta=0$. 
\end{enumerate}
\end{lemma}

\emph{proof:}
Let us fix $y\in\mathbb{N}$ positive.
Properties 1 and 2 are true since $p_x(\beta|y)$ is a proper distribution, by definition.
Let us focus on condition 3.

The first requirement is uniform convergence of the posterior $p_x(\beta|y) = p_x(y|\beta)\pi_\beta(\beta)/m(y)$.
Let us check the behavior of each component in the limit of $x\to\pm\infty$.

Let us fix $a>0$, we are interested in studying the limit for $x\to \pm \infty$, given $|\beta|\ge a$. 
From Lemma \ref{Lemma:likelihood_0_ap}, we have that for any $\beta\neq0$,
$$\lim_{x\to\pm\infty}p_x(y|\beta) = 0.$$

Let us now consider the marginal $m(y)$ and study its asymptotic behavior. By applying the change of variable $z=x\beta$ we have that
$$m(y) = \int_{-\infty}^\infty p_x(y|\beta)\pi_\beta(\beta)d\beta = \frac{1}{x}\int_{-\infty}^\infty p(y|z)\pi_\beta\Big(\frac{z}{x}\Big)dz$$
where $p(y|z) = \int_0^\infty\text{Pois}(y|e^zu)\pi(u) du$. 
Then we can observe that $m(y)$ is asymptotically equivalent to $\big(\pi_\beta(0)/x\big) \int_{-\infty}^\infty p(y|z)dz$: 
\begin{align*}
  \lim_{x\to\pm\infty}\frac{m(y)}{\big(\pi_\beta(0)/x\big) \int_{-\infty}^\infty p(y|z)dz } &= \lim_{x\to\pm\infty} \frac{\int_{-\infty}^\infty p(y|z)\pi_\beta(z/x)dz}{\pi_\beta(0) \int_{-\infty}^\infty p(y|z)dz } \\
  &\overset{DCT}{=} \frac{\int_{-\infty}^\infty  \lim_{x\to\pm\infty} p(y|z)\pi_\beta(z/x)dz}{\pi_\beta(0) \int_{-\infty}^\infty p(y|z)dz } = 1,  
\end{align*}
indicating that the marginal decays to 0 at a rate $\frac{1}{x}$ slower than the likelihood.
Note that here we can apply DCT since, by assumption the prior $\pi_\beta(z/x)$ is bounded, meaning that there exists $M_1>0:\pi_\beta(z/x)<M_1$. This implies that $f(z) = M_1p(y|z)\ge p(y|z)\pi_\beta(z/x) \forall z\in\mathbb{R}$.
Further, from the existence of joint moments, the mixed Poisson likelihood is integrable: $\int_0^\infty p(y|z)dz<\infty$, that is $f(z)$ is integrable.

If we then consider the posterior, by replacing $m(y)$ with its asymptotic equivalent we have that, when $|\beta|>0$,
$$\lim_{x\to\pm\infty} p_x(\beta|y) = \lim_{x\to\pm\infty} \frac{p_x(y|\beta)\pi_\beta(\beta)}{\pi_\beta(0)/x \int_{-\infty}^\infty p(y|z)dz }= 0$$
since $\frac{\pi_\beta(\beta)}{\pi_\beta(0) \int_{-\infty}^\infty p(y|z)dz }$ is a constant w.r.t. $x$ and 
$$xp_x(y|\beta) = \int_0^\infty x (e^{x\beta}u)^y\exp(-e^{x\beta}u)/y! \pi(u)du\to 0 \text{ as } x\to\pm\infty.$$
Where the latter limit holds for any mixing distribution $\pi$ with tails behavior as in equations \ref{eq:left_tail_ap}, \ref{eq:right_tail_ap}, as demonstrated in Lemma~\ref{Lemma:xlikelihood_0}.

We now need to verify the second requirement. Given $a>0$, applying DCT and the property above:
\begin{align*}
    \int_{\{|\beta|\ge a\}} p_x(\beta|y)d\beta &= \frac{ \int_{\{|\beta|\ge a\}} p_x(y|\beta)\pi_\beta(\beta)d\beta}{ \int_\mathbb{R} p_x(y|\beta)\pi_\beta(\beta)d\beta } \\
    &\overset{z=x\beta}{=}  \frac{\frac{1}{x} \int_{\{|z|\ge ax\}} p(y|z)\pi_\beta(z/x)dz}{ \frac{1}{x}\int_\mathbb{R} p(y|z)\pi_\beta(z/x)dz }\\
    &\overset{\text{DCT}}{\to} \frac{\pi_\beta(0)}{\pi_\beta(0)} \frac{\int_{|z|>\infty} p(y|z)dz}{\int_\mathbb{R} p(y|z)dz}=0 \text{ as } x\to\pm\infty
\end{align*}
since $\{z: |z|>a|x|\to\infty , \text{ as } x\to\pm\infty\}=\o$  and in the last step we can again apply DCT as the integrand is bounded by a function of $z$, since the prior $\pi_\beta$ is bounded.
$\square$

\vspace{1cm}

Note that the fact that the posterior satisfies property 3 means that $\forall a>0 \ Pr(|\beta|>a|y)\to0 \text{ as } x\to\pm\infty$, which is generally sufficient to prove posterior convergence to $\delta_0$. However, we will still provide below a formal proof for such convergence, applying the definition of Dirac delta.

\begin{theorem}\label{Th:pointmass_post_ap}
Under a mixed Poisson count regression model framework,  where the mixing distribution has LRV, or lighter, right tails and polynomial, or lighter, divergence at zero, we have that, for $y\in\mathbb{N}$ positive:
$$p_x(\beta|y)\to\delta_0 \text{ as } x\to\pm\infty.$$ 

\end{theorem}

\emph{proof:}
By definition, $p_x(\beta|y)$ converges in the sense of distribution to Dirac delta, iff 
$$\int_{-\infty}^\infty \phi(\beta)p_x(\beta|y)d\beta\to\phi(0) \text{ as } |x|\to\infty$$
for any continuous and bounded function $\phi\in C$.
So, fixed $\epsilon>0$, we want to show that $\exists N>0: |x|>N\implies|\int_{-\infty}^\infty \phi(\beta)p_x(\beta|y)d\beta-\phi(0)|<\epsilon$.

From Lemma \ref{Th:pointmass_post_suffconds}, property 2 (unit mass) we have that $\phi(0) \overset{p2}{=} \phi(0) \int_{-\infty}^\infty p_x(\beta|y)d\beta$, then
$$\Big| \int_{-\infty}^\infty \phi(\beta)p_x(\beta|y)d\beta - \phi(0)\Big|\le \int_{-\infty}^\infty p_x(\beta|y) |\phi(\beta)-\phi(0)| d\beta.$$

We now use the fact that $\phi$ is continuous, hence $\exists \delta>0:|\beta|<\delta\implies|\phi(\beta)-\phi(0)|<\epsilon/2$.
Then we can split the integral so that:
\begin{align*}
    \int_{-\infty}^\infty p_x(\beta|y) |\phi(\beta)-\phi(0)| d\beta &= \int_{-\delta}^\delta p_x(\beta|y) |\phi(\beta)-\phi(0)| d\beta \\
    &+ \int_{\{|\beta|>\delta\}} p_x(\beta|y) |\phi(\beta)-\phi(0)| d\beta.
\end{align*}

For the integral in the $\delta$ surrounding of 0 we have that, by continuity and property 2: 
$$\int_{-\delta}^\delta p_x(\beta|y) |\phi(\beta)-\phi(0)| d\beta < \epsilon/2 \int_{-\delta}^\delta p_x(\beta|y)  d\beta\le\epsilon/2.$$

Let us now consider the integral in the tails.
Given that $\phi$ is bounded, we have that $|\phi(\beta)|<K\implies |\phi(\beta)-\phi(0)|\le |\phi(\beta)|+|\phi(0)| <2K$.
Then $\int_{\{|\beta|>\delta\}} p_x(\beta|y) |\phi(\beta)-\phi(0)| d\beta\le 2K \int_{\{|\beta|>\delta\}} p_x(\beta|y) d\beta$.
We now use property 3 from Lemma \ref{Th:pointmass_post_suffconds}: \\ $\lim_{|x|\to \infty} \int_{\{|\beta|>\delta\}} p_x(\beta|y) d\beta=0$. Then given $|\beta|>\delta$, we can choose $N$ such that:
$$|x|>N\implies \int_{\{|\beta|>\delta\}} p_x(\beta|y) d\beta<\epsilon/(4K)$$
which results in 
$$\int_{\{|\beta|>\delta\}} p_x(\beta|y) |\phi(\beta)-\phi(0)| d\beta<\epsilon/2.$$

To conclude, we have just shown that 
$$|x|>N\implies \Big|\int_{-\infty}^\infty \phi(\beta)p_x(\beta|y)d\beta-\phi(0)\Big|<\epsilon,$$
that means $p_x(\beta|y)\to\delta_0$ as $x\to\pm\infty$. 
$\square$

\vspace{1cm}

Now, to show non-robustness to infinite covariates in the presence of multiple observations, of which only a subset is outliers, we consider two data points $d,d\out$ modeled with a mixed Poisson regression, and let $x\out\to\pm\infty$. The posterior of the regression coefficient given the two observations does not converge to the posterior given only the non-outlying observation, which is the condition required to achieve asymptotic robustness. 
The result can be extended to any higher number of observations, including any proportion of outliers. 

\begin{theorem}\label{Th:nonrobustness_infcovar_ap}
Let $y,y\out\in\mathbb{N}$, with $y\out$ positive. Under a Bayesian mixed Poisson count regression model framework,  where the mixing distribution has LRV, or lighter, right tails and polynomial, or lighter, divergence at zero, we have that 
$$\frac{p_x(\beta|y,y\out)}{p_x(\beta|y)}\nrightarrow 1 \text{ as } x\out \to\pm\infty.$$ 
\end{theorem}

$proof$:
Let us consider $\beta\neq0$.
By applying the Bayes Theorem to the joint marginal probability we have that
$$m(y,y\out) =\int_\mathbb{R} p_x(y\out|\beta)p_x(y|\beta)\pi_\beta(\beta)d\beta = \int_\mathbb{R} p_x(\beta|y\out)p_x(y|\beta)m(y\out)d\beta.$$

Given $\beta\neq0$, we have that
\begin{align*}
    \frac{p_x(\beta|y,y\out)}{p_x(\beta|y)} &= p_x(y\out|\beta) \frac{m(y)}{m(y,y\out)}\\
    &=\frac{p_x(y\out|\beta)m(y)}{\int_\mathbb{R} p_x(\beta|y\out)p_x(y|\beta)m(y\out)d\beta} \\
    & \to\frac{m(y)\pi_\beta(\beta)}{p_x(y|0)}0 = 0 \text{ as } x\out\to\pm\infty.
\end{align*}
In the last step, we use the fact that the posterior $p_x(\beta|y\out)\to\delta_0$ and $p_x(y|\beta)$ is continuous, and apply the definition of Dirac delta to get the limit of the integral at the denominator, and the fact that the likelihood for an infinite covariate outlier decays to 0 faster than its marginal, as discussed in the proof of property 3 from Lemma \ref{Th:pointmass_post_suffconds}.
$\square$

\vspace{1cm}

We have just shown that, in the presence of infinitely large covariates associated with positive counts, Poisson mixed models do not perform asymptotic outliers rejection. This is sufficient to prove that these models are not theoretically robust to outlying covariates. However, for completeness, let us now study how these models handle infinite covariates associated to 0 counts.

\begin{theorem}
    
\label{Th:posterior_zerocounts_ap}
Under a mixed Poisson count regression model framework, 
we have that, for $y=0$:
\begin{align*}
    p_x(\beta|y=0) &\to  \frac{\pi_\beta(\beta)\mathbb{I}(\beta<0)}{\int_{-\infty}^0 \pi_\beta(\beta)d\beta} \text{ as } x\to\infty,\\
    p_x(\beta|y=0) &\to \frac{\pi_\beta(\beta)\mathbb{I}(\beta>0)}{\int_0^\infty \pi_\beta(\beta)d\beta}\text{ as } x\to-\infty.
\end{align*}


\end{theorem}

\emph{proof:}
As $x$ always appears as $x\beta$ in the model, we are interested in studying the posterior limit behavior as $x\beta\to\pm\infty$.
We want to show that $p_x(\beta|y=0)\to0 \text{ as } x\beta\to\infty,$ and $p_x(\beta|y=0)\to K\propto \pi_\beta(\beta) \text{ as } x\beta\to-\infty.$
Let us first consider the likelihood:
\begin{align*}
    p_x(y=0|\beta) &= \int_0^\infty \text{Pois}(0|e^{x\beta}u)\pi(u)du\\
    &= \int_0^\infty \exp(-e^{x\beta} u)\pi(u)du\\
    &\overset{\text{DCT}}{\to} 
    \begin{cases}
        0 & \text{ if } x\beta\to+\infty\\
        1 & \text{ if } x\beta\to-\infty\\
        \int_0^\infty \exp(-u)\pi(u)du & \text{ if } \beta=0.
    \end{cases}
\end{align*}

Then, the posterior give $y=0$ is:
\begin{align*}
    p_x(\beta|y=0) &= \frac{p_x(y=0|\beta)\pi_\beta(\beta)}{\int_{-\infty}^0 p_x(y=0|\beta)\pi_\beta(\beta)d\beta + \int_0^\infty p_x(y=0|\beta)\pi_\beta(\beta)d\beta}
\end{align*}
and its limiting behavior depends on the sign of $x\beta$.
Let us now study the alternative limit cases separately.

When $x\beta\to\infty$, the likelihood decays to 0, while the marginal $m(y)$ does not. 
Then, if $\beta>0$, as $x\to\infty$ we have that $p_x(\beta|y=0)\to 0$.
Similarly, if $\beta<0$, as $x\to-\infty$ we have $p_x(\beta|y=0)\to 0$.

When $x\beta\to-\infty$, the likelihood goes to 1. 
Then, if $\beta<0$, as $x\to\infty$ we have $p_x(\beta|y=0)\to \pi_\beta(\beta)/\int_{-\infty}^0 \pi_\beta(\beta)d\beta$.
If $\beta>0$, as $x\to-\infty$ we have $p_x(\beta|y=0)\to \pi_\beta(\beta)/\int_0^\infty \pi_\beta(\beta)d\beta$.

Note that if the prior for $\beta$ is symmetric w.r.t. 0, such as for a Gaussian prior, we have that 
$p_x(\beta|y)\to 2\pi_\beta(\beta) \text{ as } x\beta\to-\infty.$
$\square$

\subsection{Outlier rejective mixed Poisson models robustness to infinite covariates}\label{appendix:mixedpoismix_rob_infcov}
Let us denote the mixed Poisson model likelihood as
$$f(y|\boldsymbol{x}'\beta) := \int_0^\infty \text{Pois}(y|\exp(\boldsymbol{x}'\beta)u)\pi(u)du$$
with $\beta\in\mathbb{R}^p: \beta\sim \pi_\beta$ (bounded and continuous, often assumed Gaussian) regression coefficients and $u>0:u\sim\pi$ mixing distribution.

We now consider the following mixture of mixed Poisson models:
$$p(y|\beta,\theta) = (1-\theta) f(y|\boldsymbol{x}'\beta) + \theta f(y|0)$$
where $\theta\in(0,1)$ is the Bernoulli parameter, capturing the proportion of type-$x$ outliers present in the data.

For simplicity, let us assume $\beta\in\mathbb{R}$ and consider two observations $(x,y)$, $(x\out,y\out) \in \mathbb{R}\times\mathbb{N}$, as done in the previous section. We then study the case when $x\out\to\pm\infty$.

\begin{theorem}\label{Th:robustness_infcovar_ap}
Let $y,y\out\in\mathbb{N}$, with $y\out$ positive. Under a Bayesian  $x$-OR mixed Poisson count regression model framework, 
we have that 
$$\frac{p_x(\beta|y,y\out,\theta)}{p_x(\beta|y,\theta)}\rightarrow 1 \text{ as } x\out \to\pm\infty.$$ 
\end{theorem}

$proof:$
\begin{align*}
  \frac{p(\beta|y,y\out,\theta)}{p(\beta|y,\theta)} &= p(y\out|\beta,\theta) \frac{p(y)}{p(y,y\out)}\\
    &= p(y\out|\beta,\theta) \frac{\int_{-\infty}^\infty p(y|\beta,\theta)\pi_\beta(\beta)d\beta}{\int_{-\infty}^\infty p(y|\beta,\theta)p(y\out|\beta,\theta)\pi_\beta(\beta)d\beta} \\
    &= [(1-\theta) f(y\out|x\out\beta) + \theta f(y\out|0)]\\
    &\hspace{.4cm} \times\frac{\int_{-\infty}^\infty p(y|\beta,\theta)\pi_\beta(\beta)d\beta}{(1-\theta)\int_{-\infty}^\infty p(y|\beta,\theta)f(y\out|x\out\beta)\pi_\beta(\beta)d\beta + \theta\int_{-\infty}^\infty p(y|\beta,\theta)f(y\out|0)\pi_\beta(\beta)d\beta} \\
    &\overset{\text{L 2.1}}{\longrightarrow} [0 + \theta f(y\out|0)] \frac{\int_{-\infty}^\infty p(y|\beta,\theta)\pi_\beta(\beta)d\beta}{\theta f(y\out|0)  \int_{-\infty}^\infty p(y|\beta,\theta)\pi_\beta(\beta)d\beta} \\
    &  \hspace{.7cm}= 1 \text{ as } x\out\to\pm\infty.
\end{align*}
where, in the third step, we substitute $p(y\out|\beta,\theta)$ with its mixture formulation, while in the fourth step, when considering the limit, we applied the DCT and Lemma \ref{Lemma:likelihood_0_ap} which states that mixed Poisson likelihoods of the form of $f(y|x\beta)$ go to 0 as $x\to\pm\infty$.
$\square$

Let us now show that when $y\out=0$ mixed Poisson mixture models are still robust to infinite covariates.

\begin{theorem}\label{Th:robustness_infcovar_zerocount}
Let $y,y\out\in\mathbb{N}_0$, with $y\out=0$. Under a Bayesian mixed Poisson mixture count regression model framework, 
we have that 
$$\frac{p_x(\beta|y,y\out=0,\theta)}{p_x(\beta|y,\theta)}\rightarrow 1 \text{ as } x\out \to\pm\infty.$$ 
\end{theorem}

$proof:$
We consider the following mixed Poisson mixture likelihood models:
$$p(y|\beta,\theta) = (1-\theta) f(y|x'\beta) + \theta f(y|0),$$
with $f(y|x'\beta) := \int_0^\infty \text{Pois}(y|\exp({x}'\beta)u)\pi(u)du.$
We then have that:
\begin{align*}
    \frac{p(\beta|y,y\out=0,\theta)}{p(\beta|y,\theta)} &= p(y\out=0|\beta,\theta) \frac{p(y)}{p(y,y\out=0)}\\
    &= p(y\out=0|\beta,\theta) \frac{\int_{-\infty}^\infty p(y|\beta,\theta)\pi_\beta(\beta)d\beta}{\int_{-\infty}^\infty p(y|\beta,\theta)p(y\out=0|\beta,\theta)\pi_\beta(\beta)d\beta} \\
    &= [(1-\theta)f(y\out=0|x\out\beta) + \theta f(y\out=0|0)] \\
    & \times \frac{\int_{-\infty}^\infty p(y|\beta,\theta)\pi_\beta(\beta)d\beta}{(1-\theta)\int_{-\infty}^\infty p(y|\beta,\theta)f(y\out=0|x\out\beta)\pi_\beta(\beta)d\beta + \theta\int_{-\infty}^\infty p(y|\beta,\theta)f(y\out=0|0)\pi_\beta(\beta)d\beta} 
\end{align*}
From Theorem \ref{Th:posterior_zerocounts_ap}, we have that:
\begin{align*}
   f(y\out=0|x\out\beta) &= \int_0^\infty \text{Pois}(0|\exp({x\out}\beta)u)\pi(u)du \\
   &= \int_0^\infty \exp(-e^{x\out\beta}u)\pi(u)du \\
   &\overset{\text{DCT}}{\to} 
       \begin{cases}
        0 & \text{ if } x\beta\to+\infty\\
        1 & \text{ if } x\beta\to-\infty\\
        \int_0^\infty \exp(-u)\pi(u)du & \text{ if } \beta=0.
    \end{cases}
\end{align*}

Then, applying the above result, we have that for $x\out\beta\to+\infty$:
$$\frac{p(\beta|y,y\out)}{p(\beta|y)} \overset{\text{DCT}}{\to} \theta f(0|0) \frac{\int_{-\infty}^\infty p(y|\beta,\theta)\pi_\beta(\beta)d\beta}{\theta f(0|0)\int_{-\infty}^\infty p(y|\beta,\theta)\pi_\beta(\beta)d\beta}=1.$$
While, for $x\out\beta\to-\infty$, we have that:
$$\frac{p(\beta|y,y\out)}{p(\beta|y)} \overset{\text{DCT}}{\to} [(1-\theta) + \theta f(0|0)] \frac{\int_{-\infty}^\infty p(y|\beta,\theta)\pi_\beta(\beta)d\beta}{[(1-\theta) + \theta f(0|0)]\int_{-\infty}^\infty p(y|\beta,\theta)\pi_\beta(\beta)d\beta}=1.$$
$\square$

\subsection{Computation details for Poisson-RSB $s$-mix model in the Kara Sea case study}
Give $i=1,\dots,n$ observations and $j(i)\in\{1,2,3\}$ study campaign for each observation, the hierarchical framework of the model used in Section 5 of the main manuscript for the Poisson-RSB $s$-mix model is:
\begin{equation}\label{eq:poismix_smix}
    \begin{split}
        y_i|\lambda_i,\eta_i &\sim \text{Poisson}(\lambda_i\eta_i)\\
        \eta_i|s&\sim (1-s)\delta_1 + s \text{RSB}(0.5,0.5) \\ %
        \log(\lambda_i) &= x_i' \beta + \alpha_{j(i)}\\
        \beta&\sim \text{N}(0,10) ; \alpha\sim \text{N}(0,2); s\sim \text{Beta}(1,1) 
    \end{split}
\end{equation}
where $s$ represents the proportion of outliers, while $\pi$ denotes the different Poisson-mixing distributions.
Marginalizing the observational layer over $s$ we can recover an equivalent formulation of the model:
\begin{equation}\label{eq:poismix_smix1}
    \begin{split}
        y_i|\lambda_i,\eta_i,s &\sim s\text{Poisson}(\lambda_i\eta_i) + (1-s)\text{Poisson}(\lambda_i)\\
        \eta_i|Z_i = 1&\sim  \text{RSB}(0.5,0.5)\\ %
        \log(\lambda_i) &= x_i' \beta + \alpha_{j(i)}\\
        \beta&\sim \text{N}(0,10) ; \alpha\sim \text{N}(0,2); s\sim \text{Beta}(1,1). 
    \end{split}
\end{equation}
In this model, given a r.v. $Z_i\sim\text{Bernoulli}(s)$, $\eta_i$ is only activated when $Z_i = 1$, while when $Z_i = 0$ the observation $y_i$ follows a standard Poisson distribution. 

When implementing the model in Stan, we follow the latter formulation, but we allow sampling from the posterior of $\eta_i$ independently of $Z_i$. This may result in extremely high values for posterior samples of $\eta_i|Z_i = 0$. To ensure that we only account for the samples $\eta_i|Z_i=1$ we recover the posterior samples of $Z_i|y_i,\eta_i,s$ as follows.
From the Bayes Theorem
$$p(Z_i|y_i,\eta_i,\lambda_i,s) \propto p(y_i|\eta_i,\lambda_i,s, Z_i)p(Z_i|s)$$
that results in a Bernoulli r.v. with posterior probabilities: 
\begin{align*}
    \text{Pr}(Z_i=1|y_i,s,\lambda_i,\eta_i) & \propto s\text{Poisson}(y_i|\eta_i\lambda_i)\\
    \text{Pr}(Z_i=0|y_i,s,\lambda_i,\eta_i) & \propto (1-s)\text{Poisson}(y_i|\lambda_i).    
\end{align*}

We therefore sample $Z_i^j$ from its posterior, a Bernoulli with rate $\frac{s^j\text{Pois}(y_i|\eta_i^j\lambda_i^j)}{(1-s^j)\text{Pois}(y_i|\lambda_i^j) + s^j\text{Pois}(y_i|\eta_i^j\lambda_i^j) }$, for each MC sample $j$ and set $\eta_i^j = 1$ when $Z_i^j=0$. 


\subsection{Additional results}\label{appendix:additional_results}


\begin{table}[!h]

\caption{Percentage of experiments, run with 10 and  with 100 regular observations and one outlier, where the posterior 80\% CI included the true value of the regression coefficients.
}
\centering
\begin{tabular}[]{llrrrrrrrr}
\toprule
Regular& Model& \multicolumn{2}{c}{LHS-$y$} & \multicolumn{2}{c}{RHS-$y$}& \multicolumn{2}{c}{LHS-$x$}& \multicolumn{2}{c}{RHS-$x$}\\
 obs.&& a & b & a & b & a & b & a & b\\
\midrule
10& Poisson & 50.00 & 23.64 & 28.64 & 45.45 & 42.73 & 31.36 & 30.45 & 25.91\\
&Negative Binomial & 98.18 & 95.00 & 79.09 & 99.55 & 99.09 & 95.45 & 74.09 & 32.27\\
&Poisson-log-t & 95.00 & 95.91 & 94.55 & 96.36 & 94.09 & 95.45 & 96.36 & 77.27\\
&Poisson-RSB & 96.36 & 97.73 & 100.00 & 91.82 & 98.18 & 91.36 & 99.09 & 90.00\\
&$x$-OR Poisson-log-t & 95.00 & 95.91 & 95.00 & 98.18 & 94.55 & 97.27 & 96.36 & 96.82\\
\addlinespace
100&Poisson & 47.27 & 23.64 & 40.91 & 41.36 & 40.00 & 29.09 & 58.18 & 46.36\\
&Negative Binomial & 97.27 & 100.00 & 75.45 & 95.45 & 99.09 & 100.00 & 71.36 & 47.73\\
&Poisson-log-t & 86.82 & 90.45 & 88.64 & 91.36 & 87.27 & 90.45 & 89.09 & 93.18\\
&Poisson-RSB & 20.45 & 77.73 & 20.00 & 75.00 & 20.00 & 80.91 & 20.45 & 74.55\\
&$x$-OR Poisson-log-t & 94.09 & 96.36 & 95.00 & 97.27 & 94.09 & 95.91 & 95.00 & 96.82\\
\bottomrule
\end{tabular}
\label{tab:exp_coverage_ab}
\end{table}

\begin{figure}[ht]
    \centering
    \includegraphics[width=0.99\linewidth]{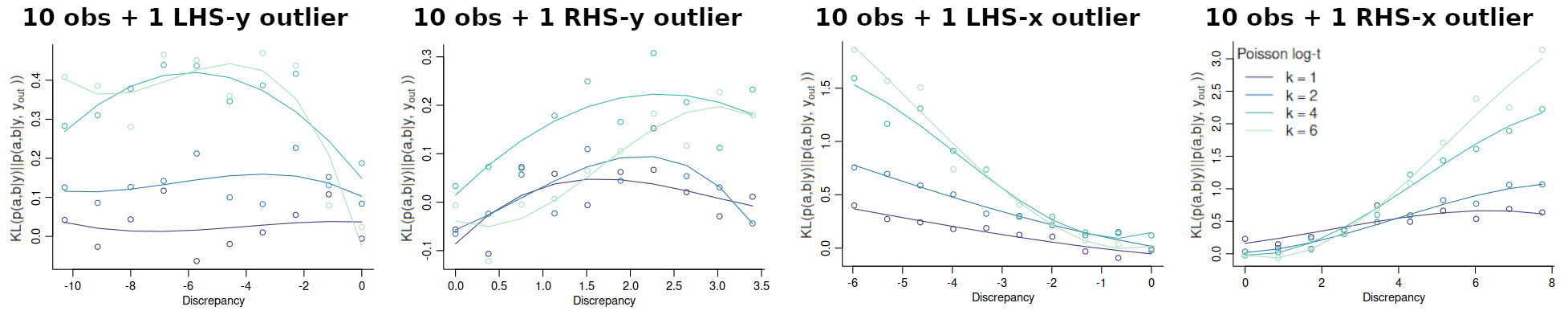}
    \caption{The effect of degrees of freedom (line and dot colors) on the behavior of Poisson-log-t mixture under different types of outliers. On the x-axis the true discrepancy and on the y-axis the KL divergence between posterior of $a,b$, with and without outlier, added to 10 regular observations. Results are shown in different plots for the four different types of outlier. A polynomial regression of third order highlights the trend in the KL divergence.}
    \label{fig:KL_poislogt_k1246}
\end{figure}

\begin{figure}[h!]
    \centering
    \includegraphics[width=.95\linewidth]{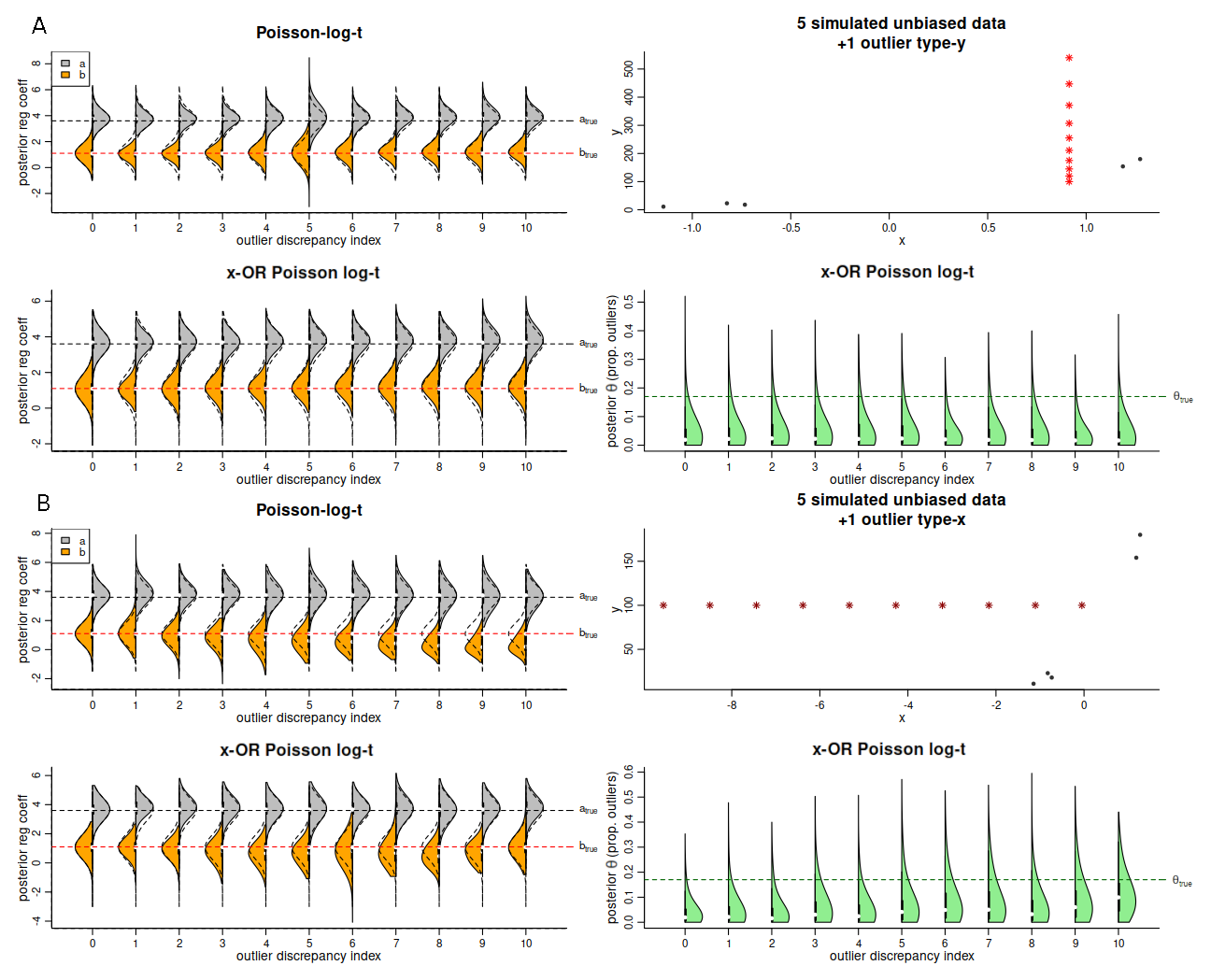}
    \caption{Posterior distributions for the intercept, $a$, and regression coefficient, $b$ (gray and orange violinlpots), and proportion of outliers, $\theta$ (green violinplots), obtained by fitting Poisson-log-t and $x$-OR Poisson-log-t to simulated datasets including five regular observations and one increasingly more extreme outlier of type L-$y$ (A) and S-$x$ (B). Note the differences in y-axis ranges.} 
    \label{fig:vioppost_yxout_extrm_logtoutmix}
\end{figure}

\begin{figure}[!ht]
    \centering
    \includegraphics[width=0.95\linewidth]{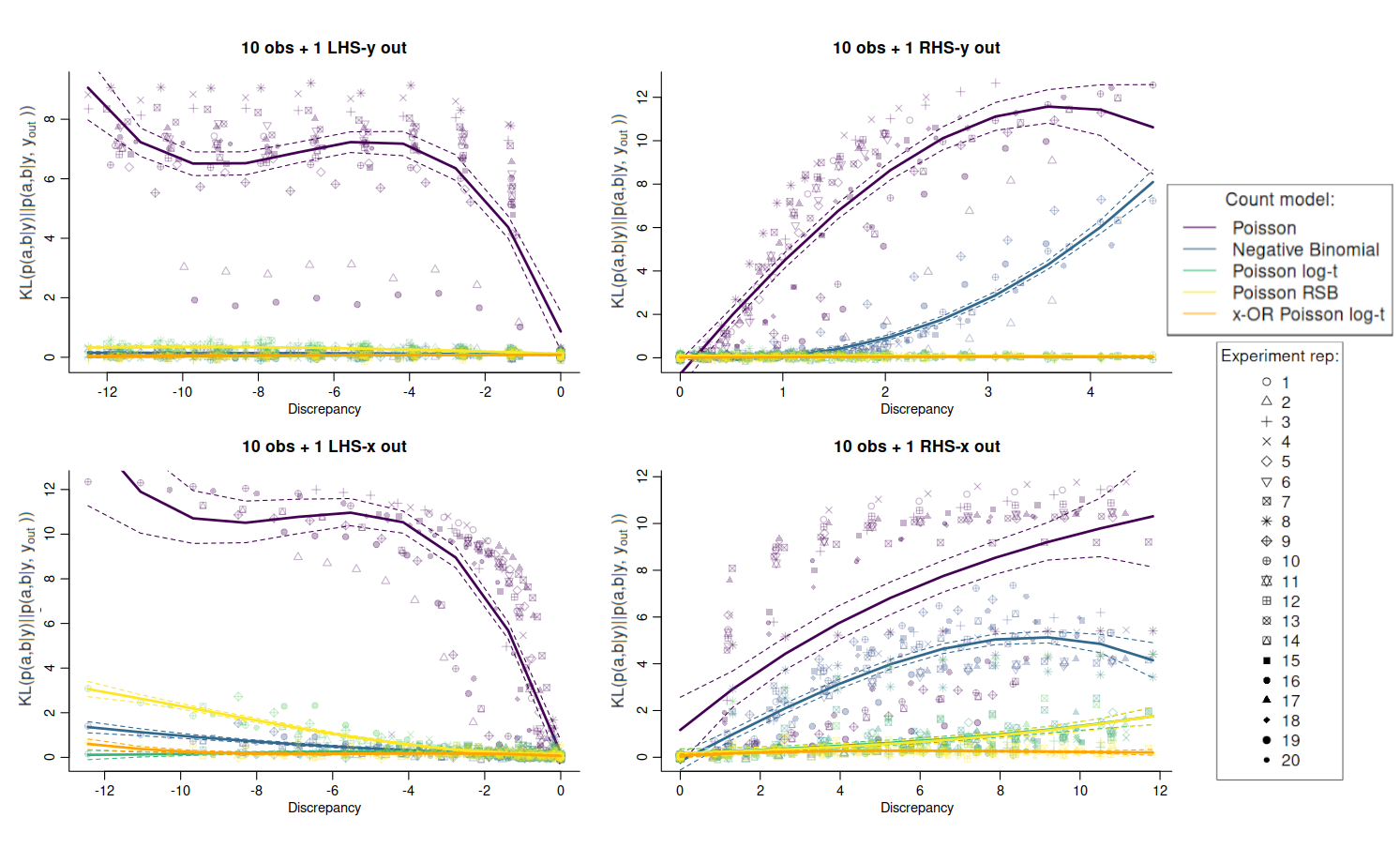} 
    \caption{ Plots of true discrepancy vs KL divergence between posteriors of the regression coefficients $a,b$ with and without outlier. Results are shown in different plots for the four different types of outlier, added to 10 regular observations. Different colors correspond to different count models, different point shapes to different simulation repetitions. We fitted a Bayesian polynomial regression of third order to highlight the trend in the KL divergence over different experiment repetitions. The lines represent the expected posterior KL and 95\% CI (dashed lines).
    }
    \label{fig:result_exp1_10obs}
\end{figure}


\bibliographystyle{ba}
\bibliography{reference}

\end{document}